\documentclass[twocolumn,aps,nofootinbib,superscriptaddress]{revtex4}

\usepackage{dsfont}
\usepackage{amsmath}
\usepackage{amssymb}
\usepackage{amsthm}
\usepackage{graphicx}
\usepackage{mathrsfs}
\usepackage{units}

\newcommand{\comment}[1]{}
\newcommand{\abs}[1]{\ensuremath{|#1|}}

\newcommand{\norm}[2]{\ensuremath{|\!|#1|\!|_{#2}}}
\newcommand{\normbig}[2]{\ensuremath{\big|\!\big|#1\big|\!\big|_{#2}}}
\newcommand{\normBig}[2]{\ensuremath{\Big|\!\Big|#1\Big|\!\Big|_{#2}}}

\newcommand{\tr}{\textnormal{tr}}
\newcommand{\trace}[1]{\ensuremath{\tr (#1)}}
\newcommand{\Trace}[1]{\ensuremath{\tr \left( #1 \right)}}

\newcommand{\ptr}[1]{\textnormal{tr}_{\textnormal{\tiny #1}}}

\newcommand{\supp}[1]{\textnormal{supp}\, \{ #1 \}}
\newcommand{\rank}{\textnormal{rank}}
\newcommand{\GF}{\mathrm{GF}}
\newcommand{\polylog}{\textnormal{polylog}}

\newcommand{\idx}[2]{{#1}_{\textnormal{\tiny #2}}}

\newcommand{\ket}[1]{| #1 \rangle}
\newcommand{\keti}[2]{| #1 \rangle_{\textnormal{\tiny #2}}}

\newcommand{\bracketi}[4]{\langle #1 | #2 | #3 \rangle_{\textnormal{\tiny #4}}}
\newcommand{\proj}[2]{| #1 \rangle\!\langle #2 |}
\newcommand{\proji}[3]{| #1 \rangle\!\langle #2 |_{\textnormal{\tiny #3}}}

\newcommand{\dpure}[2]{P(#1, #2)}
\newcommand{\Dpure}[2]{P\left( #1, #2 \right)}

\newcommand{\genfid}[2]{\bar{F}(#1, #2)}
\newcommand{\duni}[3]{\Delta(#1 | #2)_{#3}}

\newcommand{\iso}{\cong}
\newcommand{\kron}{\otimes}
\newcommand{\eps}{\varepsilon}

\newcommand{\h}{\ensuremath{\mathcal{H}}}
\newcommand{\hi}[1]{\ensuremath{\mathcal{H}_{\textnormal{\tiny #1}}}}
\newcommand{\hA}{\hi{A}}
\newcommand{\hB}{\hi{B}}
\newcommand{\hC}{\hi{C}}
\newcommand{\hD}{\hi{D}}
\newcommand{\hE}{\hi{E}}

\newcommand{\hZ}{\hi{Z}}

\newcommand{\hAB}{\hi{AB}}
\newcommand{\hAC}{\hi{AC}}

\newcommand{\hAE}{\hi{AE}}
\newcommand{\hABC}{\hi{ABC}}

\newcommand{\id}{\ensuremath{\mathds{1}}}
\newcommand{\idi}[1]{\ensuremath{\mathds{1}_{\textnormal{\tiny #1}}}}
\newcommand{\idA}{\idi{A}}
\newcommand{\idB}{\idi{B}}

\newcommand{\idE}{\idi{E}}

\newcommand{\idZ}{\idi{Z}}

\newcommand{\idAB}{\idi{AB}}

\newcommand{\opidi}[1]{\ensuremath{\mathcal{I}_{\textnormal{\tiny #1}}}}
\newcommand{\opidE}{\opidi{E}}

\newcommand{\linops}[1]{\ensuremath{\mathcal{L}(#1)}}
\newcommand{\hermops}[1]{\ensuremath{\mathcal{L}^\dagger(#1)}}
\newcommand{\posops}[1]{\ensuremath{\mathcal{P}(#1)}}

\newcommand{\normstates}[1]{\ensuremath{\mathcal{S}_{=}(#1)}}
\newcommand{\subnormstates}[1]{\ensuremath{\mathcal{S}_{\leq}(#1)}}

\newcommand{\cF}{\mathcal{F}}

\newcommand{\cH}{\mathcal{H}}

\newcommand{\cX}{\mathcal{X}}

\newcommand{\cZ}{\mathcal{Z}}

\newcommand{\rhot}{\ensuremath{\tilde{\rho}}}
\newcommand{\rhob}{\ensuremath{\bar{\rho}}}

\newcommand{\rhoA}{\ensuremath{\idx{\rho}{A}}}
\newcommand{\rhoB}{\ensuremath{\idx{\rho}{B}}}
\newcommand{\rhoC}{\ensuremath{\idx{\rho}{C}}}
\newcommand{\rhoE}{\ensuremath{\idx{\rho}{E}}}
\newcommand{\rhoAB}{\ensuremath{\idx{\rho}{AB}}}
\newcommand{\rhoAC}{\ensuremath{\idx{\rho}{AC}}}

\newcommand{\rhoXE}{\ensuremath{\idx{\rho}{XE}}}

\newcommand{\rhoXB}{\ensuremath{\idx{\rho}{XB}}}

\newcommand{\rhotAB}{\ensuremath{\idx{\rhot}{AB}}}
\newcommand{\rhobAB}{\ensuremath{\idx{\rhob}{AB}}}

\newcommand{\rhotAC}{\ensuremath{\idx{\rhot}{AC}}}

\newcommand{\rhotB}{\ensuremath{\idx{\rhot}{B}}}
\newcommand{\rhobB}{\ensuremath{\idx{\rhob}{B}}}
\newcommand{\rhobE}{\ensuremath{\idx{\rhob}{E}}}

\newcommand{\rhobXB}{\ensuremath{\idx{\rhob}{XB}}}

\newcommand{\rhoABC}{\ensuremath{\idx{\rho}{ABC}}}

\newcommand{\rhoABE}{\ensuremath{\idx{\rho}{ABE}}}

\newcommand{\rhotABC}{\ensuremath{\idx{\rhot}{ABC}}}

\newcommand{\sigmat}{\ensuremath{\tilde{\sigma}}}

\newcommand{\sigmaB}{\ensuremath{\idx{\sigma}{B}}}

\newcommand{\sigmaE}{\ensuremath{\idx{\sigma}{E}}}

\newcommand{\taut}{\ensuremath{\tilde{\tau}}}
\newcommand{\taub}{\ensuremath{\bar{\tau}}}

\newcommand{\tauB}{\ensuremath{\idx{\tau}{B}}}

\newcommand{\tauE}{\ensuremath{\idx{\tau}{E}}}
\newcommand{\tauAB}{\ensuremath{\idx{\tau}{AB}}}
\newcommand{\tauAC}{\ensuremath{\idx{\tau}{AC}}}

\newcommand{\tauABC}{\ensuremath{\idx{\tau}{ABC}}}

\newcommand{\omegaA}{\ensuremath{\idx{\omega}{A}}}

\newcommand{\omegaZ}{\ensuremath{\idx{\omega}{Z}}}

\newcommand{\phit}{\ensuremath{\tilde{\phi}}}

\newcommand{\phiAB}{\ensuremath{\idx{\phi}{AB}}}

\newcommand{\muB}{\ensuremath{\idx{\mu}{B}}}

\newcommand{\DeltaB}{\ensuremath{\idx{\Delta}{B}}}

\newcommand{\PiA}{\ensuremath{\idx{\Pi}{A}}}
\newcommand{\PiB}{\ensuremath{\idx{\Pi}{B}}}

\newcommand{\PiAC}{\ensuremath{\idx{\Pi}{AC}}}

\newcommand{\GammaB}{\ensuremath{\idx{\Gamma}{B}}}

\newcommand{\hh}[4]{\ensuremath{H_{#1}^{#2}(\textnormal{#3})_{#4}}}
\newcommand{\chh}[5]{\ensuremath{H_{#1}^{#2}(\textnormal{#3}|\textnormal{#4})_{#5}}}

\newcommand{\chhalt}[5]{\ensuremath{\widehat{H}_{#1}^{#2}(\textnormal{#3}|\textnormal{#4})_{#5}}}

\newcommand{\chhrel}[5]{\ensuremath{h_{#1}^{#2}(\textnormal{#3}|\textnormal{#4})_{#5}}}

\newcommand{\hmin}[2]{\hh{\textnormal{min}}{}{#1}{#2}} % min entropy

\newcommand{\chmin}[3]{\chh{\textnormal{min}}{}{#1}{#2}{#3}}
\newcommand{\chmineps}[3]{\chh{\textnormal{min}}{\eps}{#1}{#2}{#3}}
\newcommand{\chmineeps}[4]{\chh{\textnormal{min}}{#1}{#2}{#3}{#4}}
\newcommand{\chminalt}[3]{\chhalt{\textnormal{min}}{}{#1}{#2}{#3}} % alternative min entropy
\newcommand{\chminalteps}[3]{\chhalt{\textnormal{min}}{\eps}{#1}{#2}{#3}}
\newcommand{\chminalteeps}[4]{\chhalt{\textnormal{min}}{#1}{#2}{#3}{#4}}
\newcommand{\chminrel}[3]{\chhrel{\textnormal{min}}{}{#1}{#2}{#3}}
\newcommand{\chmax}[3]{\chh{\textnormal{max}}{}{#1}{#2}{#3}}
\newcommand{\chmaxeps}[3]{\chh{\textnormal{max}}{\eps}{#1}{#2}{#3}}
\newcommand{\chmaxeeps}[4]{\chh{\textnormal{max}}{#1}{#2}{#3}{#4}}
\newcommand{\chmaxalt}[3]{\chhalt{\textnormal{max}}{}{#1}{#2}{#3}}  % alternative max entropy
\newcommand{\chmaxalteps}[3]{\chhalt{\textnormal{max}}{\eps}{#1}{#2}{#3}}
\newcommand{\chmaxalteeps}[4]{\chhalt{\textnormal{max}}{#1}{#2}{#3}{#4}}
\newcommand{\chvn}[3]{\chh{}{}{#1}{#2}{#3}}
\newcommand{\pcguess}[2]{p_{\textrm{guess}}(\textnormal{#1}|\textnormal{#2})}

\newcommand{\epsball}[1]{\ensuremath{\mathcal{B}^\eps(#1)}}
\newcommand{\epsballpure}[1]{\ensuremath{\mathcal{B}_\textnormal{p}^\eps(#1)}}

\newcommand{\eepsball}[2]{\ensuremath{\mathcal{B}^{#1}(#2)}}

\theoremstyle{plain}
\newtheorem{lemma}{Lemma}
\newtheorem{theorem}[lemma]{Theorem}
\newtheorem{corollary}[lemma]{Corollary}

\theoremstyle{definition}
\newtheorem{definition}{Definition}

\begin{document}

\title{Leftover Hashing Against Quantum Side Information}

\date{\today}

\author{Marco \surname{Tomamichel}}
\email[]{marcoto@phys.ethz.ch}
\affiliation{Institute for Theoretical Physics, ETH Zurich, 8093
 Zurich, Switzerland.}
\author{Christian \surname{Schaffner}}
\email[]{c.schaffner@cwi.nl}
\affiliation{Centrum Wiskunde \& Informatica (CWI), Amsterdam, 
The Netherlands.}
\author{Adam \surname{Smith}}
\email[]{asmith@cse.psu.edu}
\affiliation{Pennsylvania State University,
University Park, PA 16802, USA.}
\author{Renato \surname{Renner}}
\email[]{renner@phys.ethz.ch}
\affiliation{Institute for Theoretical Physics, ETH Zurich, 8093
 Zurich, Switzerland.}

\begin{abstract}
  The Leftover Hash Lemma states that the output of a two-universal
  hash function applied to an input with sufficiently high entropy is
  almost uniformly random. In its standard formulation, the lemma
  refers to a notion of randomness that is (usually implicitly)
  defined with respect to classical side information. Here, we prove a
  (strictly) more general version of the Leftover Hash Lemma that is
  valid even if side information is represented by the state of a
  quantum system. Furthermore, our result applies to arbitrary
  $\delta$-almost two-universal families of hash functions. The
  generalized Leftover Hash Lemma has applications in cryptography,
  e.g., for key agreement in the presence of an adversary who is not
  restricted to classical information processing.
\end{abstract}

\maketitle

\section{Introduction}
\label{sec:intro}

We will first consider the task of extracting uniform randomness from
a random variable and introduce the Leftover Hash Lemma.  Following
its discussion, we extend the scenario to include side information
that is potentially stored in a quantum state.

\subsection{Randomness Extraction}

Consider a random variable $X$ that is partially known to an agent,
i.e., the agent possesses side information $E$ correlated to $X$. One
may ask whether it is possible to extract from $X$ a part $Z$ that is
completely unknown to the agent, i.e., uniform conditioned on $E$. If
yes, what is the maximum size of $Z$? And how is $Z$ computed?

The \emph{Leftover Hash Lemma} answers these questions. It states that
extraction of uniform randomness $Z$ is possible whenever the agent's
uncertainty about $X$ is sufficiently large. More precisely, the
number $\ell$ of extractable bits is approximately equal to the
\emph{min-entropy of $X$ conditioned on $E$}, denoted $H_{\min}(X|E)$
(see Section~\ref{sec:intro/side} for a definition and properties).
Furthermore, $Z$ can be computed as the output of a function $f$
selected at random from a suitably chosen family of functions $\cF$,
called \emph{two-universal family of hash functions} (see
Section~\ref{sec:intro/two} for a definition).  Remarkably, the family
$\cF$ can be chosen without knowing the actual probability
distribution of $X$ and only depends on the alphabet $\cX$ of $X$ and
the number of bits $\ell$ to be extracted.

\begin{lemma}[Classical Leftover Hash Lemma] \label{lemma:class} Let $X$
  and $E$ be random variables and let $\cF$ be a two-universal family
  of hash functions with domain $\cX$ and range $\{0,1\}^\ell$. Then, on
  average over the choices of $f$ from $\cF$, the distribution of the
  output $Z := f(X)$ is $\Delta$-close from uniform conditioned on
  $E$\footnote{The distance from uniform $\Delta$ measures the statistical distance
    of the probability distribution of $X$ given E to a uniform distribution. See 
    Section~\ref{sec:proof} for a formal definition.},
  where
  \begin{align*}
    \Delta = \frac{1}{2} \sqrt{2^{\ell - H_{\min}(X|E)}} \ .
  \end{align*}
\end{lemma}

The lemma immediately implies that for a \emph{fixed} joint
distribution of $X$ and $E$, there is a \emph{fixed} function $f$ that
extracts almost uniform randomness. More precisely, given any $\Delta
> 0$, there exists a function $f$ that produces
\begin{align} \label{eq:lowerbound}
  \ell = \Big\lfloor H_{\min}(X|E) - 2\, \log \frac{1}{2\Delta} \Big\rfloor
\end{align}
bits that are $\Delta$-close to uniform and independent of
$E$.\footnote{We use $\log$ to denote the binary logarithm.}

The Leftover Hash Lemma plays an important role in a variety of
applications in computer science and cryptography (see, e.g.,
\cite{stinson02} for an overview). A prominent example is
\emph{privacy amplification}, i.e., the task of transforming a weakly
secret key (over which an adversary may have partial knowledge $E$),
into a highly secret key (that is uniform and independent of the
adversary's information $E$). It was in this context that the use of
two-universal hashing for randomness distillation has first been
proposed~\cite{bennett88}. Originally, the analysis was however
restricted to situations where $X$ is uniform and $E$ is bounded in
size.  Later, versions of the Leftover Hash Lemma similar to
Lemma~\ref{lemma:class} above have been proved independently
in~\cite{impagliazzo89} and~\cite{bennett95}.  The term \emph{leftover
  hashing} was coined in~\cite{zuckerman89}, where its use for recycling
the randomness in randomized algorithms and for the construction of
pseudo-random number generators is discussed (see
also~\cite{impagliazzo89,hastad99}).

\subsection{Quantum Side Information}
\label{sec:intro/side}

A majority of the original work on universal hashing is based entirely
on probability theory and side information is therefore (often
implicitly) assumed to be represented by a \emph{classical} system $E$
(modeled as a random variable).\footnote{If the side information $E$
  is classical, the Leftover Hash Lemma can be formulated without the
  need to introduce $E$ explicitly (see, e.g.,
  \cite{impagliazzo89}). Instead, one may simply interpret all probability
  distributions as being conditioned on a fixed value of the side
  information.}  In fact, since hashing is an entirely ``classical''
process (a simple mapping from a random variable $X$ to another random
variable $Z$), one may expect that the physical nature of the side
information is irrelevant and that a purely classical treatment is
sufficient. This is, however, not necessarily the case. It has been
shown, for instance, that the output of certain extractor functions
may be partially known if side information about their input is stored
in a \emph{quantum} device of a certain size, while the same output is
almost uniform conditioned on any side information stored in a
\emph{classical} system of the same size (see~\cite{gavinsky07} for a
concrete example and~\cite{koenig07} for a more general
discussion).\footnote{Note that there is no sensible notion of a
  conditional probability distribution where the conditioning is on
  the state of a \emph{quantum} (as opposed to a \emph{classical})
  system.  An implicit treatment of side information $E$, where one
  considers all probability distributions to be conditioned on a
  specific value of $E$, as explained in the previous footnote, is
  therefore not possible in the general case. }

Here, we follow a line of research started
in~\cite{maurer05,renner05,rennerkoenig05} and study randomness
extraction in the presence of quantum side information $E$ (which, of
course, includes situations where $E$ is partially or fully
classical.) More specifically, our goal is to establish a generalized
version of Lemma~\ref{lemma:class} which holds if the system $E$ is
quantum-mechanical. For this, we first need to quickly review the
notion of \emph{min-entropy} as well as of the notion of
\emph{uniformity}, which need to be extended accordingly.

The definition of \emph{uniformity} in the context of quantum side
information $E$ is rather straightforward. Let $Z$ be a classical
random variable which takes any value $z \in \cZ$ with probability
$p_z$ and let $E$ be a quantum system whose state conditioned on $Z=z$
is given by a density operator $\rhoE^{[z]}$ on $\cH_E$. This
situation is compactly described by the \emph{classical-quantum} (CQ) state
\begin{align} \label{eq:cqstate}
  \idx{\rho}{ZE} := \sum_{z \in \cZ} p_z\, \proji{z}{z}{Z} \kron \rhoE^{[z]} \ ,
\end{align}
defined on the product space $\hZ \kron \hE$, where $\hZ$ is a Hilbert
space with orthonormal basis $\{\keti{z}{Z}\}_{z \in \cZ}$. We say
that \emph{$Z$ is uniform conditioned on $E$} if $\idx{\rho}{ZE}$ has
product form $\omegaZ \kron \rhoE$, where $\omegaZ := \idZ/
|\cZ|$ is the maximally mixed state on $\cH_Z$. More generally, we say
that $Z$ is \emph{$\Delta$-close to uniform conditioned on $E$} if
there exists a state $\sigmaE$ on $E$ for which the
trace distance between $\idx{\rho}{ZE}$ and $\omegaZ \kron
\sigmaE$ is at most $\Delta$ (see Section~\ref{sec:proof} for a
formal definition). The trace distance is a natural choice of
metric because it corresponds to the \emph{distinguishing
  advantage}.\footnote{Let $p_{\mathrm{succ}}$ be
  the maximum probability that a distinguisher, presented with a
  random choice of either the state $\rho$ or the state $\sigma$, can
  correctly guess which of the two he has seen. The
  \emph{distinguishing advantage} is then defined as the advantage
  compared to a random guess, which is given by $p_{\mathrm{succ}} -
  \frac{1}{2} = \frac{1}{4} \| \rho - \sigma\|_1$ (see e.g.~\cite{nielsen00})} Furthermore, in
the purely classical case, the trace distance reduces to the
statistical distance. 

Next, we generalize the notion of min-entropy to situations involving
quantum side information. Before we do this, note that the classical min-entropy
$\chmin{X}{E}{}$ has an operational interpretation as the guessing
probability of X given E, namely
\begin{align}
  \label{eqn:guessing-prob}
  \chmin{X}{E}{} &= -\log \pcguess{X}{E} \, .
\end{align}
Here, $p_{\textrm{guess}}(X|E)$ denotes the probability of correctly
guessing the value of $X$ using the optimal strategy with access to
$E$. The optimal strategy in the classical case is to guess, for each
value of $e$ of $E$, the $X$ with the highest conditional probability
$P_{X|E=e}$. The guessing probability is thus
\begin{align*}
  \pcguess{X}{E} = \sum_e P_E(e)\, \max_x P_{X|E=e}(x) \, . 
\end{align*}
A generalization of the min-entropy to situations where $E$ may be a
quantum system has first been proposed in~\cite{renner05} (see
Section~\ref{sec:def} for a formal definition). As shown
in~\cite{koenig08}, the operational
interpretation~\eqref{eqn:guessing-prob} naturally extends to this
more general case. In other words, the min-entropy, $\chmin{X}{E}{}$,
is a measure for the probability of guessing $X$ using an optimal
strategy with access to the quantum system $E$.

However, the actual requirement on the entropy measure used in
Lemma~\ref{lemma:class} is that it accurately characterizes the total
amount of randomness contained in $X$, i.e.\ the number of uniformly
random bits that can be extracted using an optimal extraction
strategy. As we will show below, $\chmin{X}{E}{}$ (or, more precisely,
a smooth version of it) meets this requirement.

%The expression for the conditional
%min-entropy~\eqref{eqn:guessing-prob} naturally extends to the case
%where the side information $E$ is a quantum state and the optimal
%guessing strategy is potentially a quantum strategy. The quantum
%conditional min-entropy (see also Section~\ref{sec:def} for a formal
%definition) was first proposed in~\cite{renner05} and its operational
%interpretation as a guessing probability was established
%in~\cite{koenig08}.

%Its interpretation as a guessing probability already indicates that the
%quantum conditional min-entropy might be used to extend
%Lemma~\ref{lemma:class} to the quantum domain. However, in
%order to convincingly argue that $\chmin{X}{E}{}$ is indeed the right
%quantity to characterize randomness extraction, we need to show that
%it accurately quantifies the total amount of uniform randomness that
%can be extracted from $X$.

For this purpose, let $\idx{\rho}{XE}$ be fixed and assume that $f$ is a function
that maps $X$ to a string $Z = f(X) \in \{0,1\}^\ell$ of length $\ell$ that
is uniform conditioned on the side information $E$. Then, obviously,
the probability of guessing $Z$ correctly given $E$ is equal to
$2^{-\ell}$ and, by virtue of~\eqref{eqn:guessing-prob}, we find that
\begin{align}
  \label{eq:Hminm}
  \chmin{Z}{E}{} = \ell \, .
\end{align}
Furthermore, the probability of guessing $Z = f(X)$ correctly cannot
be smaller than the probability of guessing $X$, correctly. This fact
can again be expressed in terms of min-entropies, 
\begin{align}
  \label{eq:Hminf}
  \chmin{Z}{E}{} \leq \chmin{X}{E}{}  \, ,
\end{align}
i.e., the min-entropy can only decrease under the action of a
function. Combining~\eqref{eq:Hminm} and~\eqref{eq:Hminf} immediately
yields
\begin{align} \label{eq:Hmintight}
  \ell \leq \chmin{X}{E}{} \, .
\end{align}
We conclude that the number $\ell$ of uniform bits (relative to
$E$) that can be extracted from data $X$ is upper
bounded by the min-entropy of $X$ conditioned on $E$. This result may
be seen as a converse of~\eqref{eq:lowerbound}.

So far, the claim~\eqref{eq:Hmintight} is restricted to the extraction
of \emph{perfectly uniform} randomness. In order to extend this
concept to the more general case of approximately uniform randomness,
we need to introduce the notion of \emph{smooth} min-entropy. Roughly
speaking, for any $\eps \geq 0$, the \emph{$\eps$-smooth min-entropy
  of $X$ given $E$}, denoted $\chmineps{X}{E}{}$, is defined as the
maximum value of $\chmin{X}{E}{}$ evaluated for all density operators
$\rhot$ that are $\eps$-close to $\rho$ (see
Section~\ref{sec:smoothentropies} for a formal definition).

The above argument leading to~\eqref{eq:Hmintight} can be generalized
in a straightforward manner to smooth min-entropy, and results in the
bound
\begin{align*}
  \ell \leq \chmineeps{2\sqrt{\Delta}}{X}{E}{} 
\end{align*}
for the maximum number $\ell$ of extractable bits that are $\Delta$-close
to uniform conditioned on $E$.  Crucially, our
extended version of the Leftover Hash Lemma implies that this bound
can be reached, up to additive terms of order $\log(1/\Delta)$ (see
Theorem~\ref{thm:two-universal-hashing} and
Theorem~\ref{thm:almost-two-universal-hashing}). We thus conclude that
the min-entropy of $X$ conditioned on $E$, in particular its
``smoothed'' version, is an accurate measure for the amount of uniform
randomness (conditioned on $E$) that can be extracted from $X$.

\subsection{Almost Two-Universal Hashing}
\label{sec:intro/two}

The notion of two-universal hashing has been introduced by Carter and
Wegman~\cite{carter79}.  A family $\cF$ of functions from $\cX$ to
$\cZ$ is said to be \emph{two-universal} if, for any pair of distinct
inputs $x$ and $x'$, and for $f$ chosen at random from $\cF$, the
probability of a \emph{collision} $f(x) = f(x')$ is not larger than
$\delta := 1/|\cZ|$. Note that this value for the collision
probability corresponds to the one obtained by choosing $\cF$ as the
family of \emph{all} functions with domain $\cX$ and range $\cZ$.

Later, the concept of two-universal hashing has been generalized to
arbitrary collision probabilities $\delta$~\cite{stinson94}.
Namely, a family of functions $\cF$ from
$\cX$ to $\cZ$ is called \emph{$\delta$-almost two-universal} if 
\begin{align}
  \label{eqn:almost-two-universal}
  \Pr_{f \in \cF} \left[ f(x) = f(x') \right] \leq \delta
\end{align}
for any $x \neq x'$. A two-universal family as above simply
corresponds to the special case $\delta = 1/|\cZ|$.

The classical Leftover Hash Lemma (Lemma~\ref{lemma:class}) can be
generalized to $\delta$-almost two-universal hash
functions~\cite{stinson02}. More precisely, when extracting an
$\ell$-bit string from data $X$, its distance from uniform conditioned
on $E$ is bounded by $\Delta = \frac{1}{2} \sqrt{(2^{\ell}\delta - 1)
  + 2^{\ell -H_{\min}(X|E)}}$.

\subsection{Main result}

Our main result is a generalization of the Leftover Hash Lemma for
$\delta$-almost two-universal families of hash functions which is
valid in the presence of quantum side information. While the statement
is new for general $\delta$-almost two-universal hash functions, the
special case where $\delta = 2^{-\ell}$ has been proved previously by
one of us~\cite{renner05}.

\begin{lemma}[General Leftover Hash Lemma]
  \label{lemma:quant} Let $X$
  be a random variable, let $E$ be a quantum system, and let $\cF$ be
  a $\delta$-almost two-universal family of hash functions from
  $\cX$ to $\{0,1\}^\ell$. Then, on average over the choices of $f$ from
  $\cF$, the output $Z := f(X)$ is $\Delta$-close to uniform
  conditioned on $E$, where
  \begin{align*} \Delta = \inf_{\eps > 0} \
    \frac{1}{2} \sqrt{(2^\ell \delta - 1) + 2^{\ell - \chmin{X}{E}{} +
        \log(2/\eps^2 + 1)}} + \eps \ .
  \end{align*}
  Furthermore, if $\delta \leq 2^{-\ell}$, i.e., if $\cF$ is
  two-universal, then
  \begin{align} \label{eq:epsperfect}
    \Delta = \frac{1}{2} \sqrt{2^{\ell - \chmin{X}{E}{}}} \, .
  \end{align}
\end{lemma}

Note that inserting $\delta = 2^{-\ell}$ into the first expression for
$\Delta$ yields a formula which is less tight
than~\eqref{eq:epsperfect}. The latter, therefore, requires a separate
proof. In the technical part below, the two claims are formulated more
generally for the \emph{smooth} min-entropy
(Theorem~\ref{thm:two-universal-hashing} and
Theorem~\ref{thm:almost-two-universal-hashing}).

\subsection{Applications and Related Work}

Quantum versions of the Leftover Hash Lemma~\cite{renner05} for
two-universal families of hash functions have been used in the context
of privacy amplification against a quantum
adversary~\cite{rennerkoenig05,koenig07}. This application has gained
prominence with the rise of quantum cryptography and quantum key
distribution in particular. There, the side information $E$ is
gathered during a key agreement process between two parties by an
eavesdropper who is not necessarily limited to classical information
processing. The quantum generalization of the Leftover Hash Lemma is
then used to bound the amount of secret key that can be distilled by
the two parties.

The restriction to two-universal families of hash functions leads to
the need for a random seed of length $\Theta(n)$, where $n$ is the
length in bits of the original partially secret string. This seed is
used to choose $f$ from a two-universal family $\mathcal{F}$. The main
result of this paper, Lemma~\ref{lemma:quant}, and a suitable
construction of a $\delta$-almost two-universal family of hash
functions (see Section~\ref{sec:constructions}) allow for a shorter
seed of length proportional to $\ell$, $\log \frac{n}{\ell}$ and $\log
\frac{1}{\Delta}$. The length of secret key that can be extracted with
this method is only reduced by a term proportional to $\log
\frac{1}{\Delta}$ compared to the extractor using two-universal
hashing. Furthermore, the generalized Leftover Hashing Lemma allows
for an extension of existing cryptographic security proofs to
$\delta$-almost two-universal families of hash functions and may lead
to a speed-up in practical implementations.\footnote{See,
  e.g.~\cite{assche06} and~\cite{lodewyck07}, where a practical
  implementation of privacy amplification is discussed in Section V.}

Recently, the problem of randomness extraction with quantum side
information has generated renewed interest.  It has been shown that
the classical technique~\cite{dodis05} of XORing a classical source
about which an adversary holds quantum information with a
$\delta$-biased mask results in a uniformly distributed
string~\cite{fehr08}\footnote{See also~\cite{dupuis07} for a
  generalization of this work to the fully quantum setting.}.

However, to achieve even shorter seed lengths, more advanced
techniques such as Trevisan's~\cite{trevisan01} extractor have been
studied in~\cite{ta-shma08,de09,portmann09}. In~\cite{de09}, it is
shown that a seed of length $O(\polylog\,n)$ is sufficient to generate
a key of length $\ell \approx \hmin{X}{} - \log \dim \hE$, where $\dim
\hE$ is a measure of the size of the adversary's quantum
memory. In~\cite{portmann09}, the result was extended to the formalism
of conditional min-entropies. They attain a key length of $\ell
\approx \chmineps{X}{E}{}$, which can be arbitrarily larger than
$\hmin{X}{} - \log \dim \hE$. Furthermore, as we show
in~\eqref{eq:Hmintight}, this key length is almost optimal.  Our
result may be useful to further improve the performance of these
extractors (see discussion in~\cite{portmann09}).

Furthermore, our result should be used instead of the classical
Leftover Hashing Lemma whenever randomness is extracted in a context
governed by the laws of quantum physics. For example, consider a
device that needs a seed that is random conditioned on its internal
state. In this case the use of the \emph{classical} Leftover Hashing
Lemma instead of its quantum version, Lemma~\ref{lemma:quant},
corresponds to the implicit and potentially unjustified assumption
that the device does not make use of quantum mechanics.

\subsection{Organization of the paper}

In Section~\ref{sec:smoothentropies}, we discuss various aspects of
the smooth entropy framework, which will be needed for our proof. We
then give the proof of our generalized Leftover Hash Lemma
(Lemma~\ref{lemma:quant}) in Section~\ref{sec:proof}.  More precisely,
we provide statements of the Leftover Hashing Lemma for two-universal
and $\delta$-almost two-universal hashing in terms of the smooth
min-entropy (Theorems~\ref{thm:two-universal}
and~\ref{thm:almost-two-universal}).  Finally, in
Section~\ref{sec:constructions}, we combine known constructions of
$\delta$-almost two-universal hash functions and discuss their use for
randomness extraction with shorter random seeds.
Appendix~\ref{app:alt} may be of independent interest because it
establishes a relation between the smooth min- and max-entropies (as
defined above and used in \cite{koenig08,tomamichel08,tomamichel09})
and certain related entropic quantities used in earlier work (e.g.,
in~\cite{renner05})

\section{Smooth Entropies}
\label{sec:smoothentropies}
\label{sec:def}

Let $\h$ be a finite-dimensional Hilbert space. We use $\linops{\h}$, $\hermops{\h}$
and $\posops{\h}$ to denote the set of linear, Hermitian and positive
semi-definite operators on $\h$, respectively. We
define the set of normalized quantum states by $\normstates{\h} := \{
\rho \in \posops{\h} : \tr\,\rho = 1 \}$ and the set of sub-normalized
states by $\subnormstates{\h} := \{ \rho \in \posops{\h} : 0 <
\tr\,\rho \leq 1 \}$. Given a pure state $\ket{\phi} \in \h$, we use
$\phi = \proj{\phi}{\phi}$ to denote the corresponding projector in
$\posops{\h}$. The inverse of a Hermitian operator is meant to be
taken on its support only (generalized inverse).
Given a bipartite Hilbert space $\hAB := \hA \kron \hB$ and a state
$\rhoAB \in \subnormstates{\hAB}$, we denote by $\rhoA$ and $\rhoB$
its marginals $\rhoA = \ptr{B}\,\rhoAB$ and $\rhoB = \ptr{A}\,\rhoAB$.

The \emph{trace distance} between states $\rho$ and $\tau$ is given by
$\frac{1}{2} \norm{\rho - \tau}{1} = \frac{1}{2} \tr\,\abs{\rho -
 \tau}$. We also employ the \emph{purified distance}
$P$ as a metric on $\subnormstates{\h}$~\cite{tomamichel09}. It is an
upper bound on the trace distance and defined in terms of the
\emph{generalized fidelity} $\bar{F}$ as
\begin{align*}
\dpure{\rho}{\tau} &:= \sqrt{1 - \genfid{\rho}{\tau}^2}\, , \quad
\textrm{where} \\
\genfid{\rho}{\tau} &:= \tr \abs{\sqrt{\rho}\sqrt{\tau}} + \sqrt{(1 -
  \tr\,\rho)(1 - \tr\,\tau)}\, .
\end{align*}
We will need that the purified distance is a monotone
under trace non-increasing completely positive maps (CPMs). Let
$\mathcal{E}$ be a trace non-increasing CPM, then~\cite{tomamichel09}
\begin{align}
 \label{eqn:purified-monotone}
 \dpure{\rho}{\tau} \geq \Dpure{\mathcal{E}(\rho)}{\mathcal{E}(\tau)}
 \, .
\end{align}
Note that the projections $\rho \mapsto \Pi \rho \Pi$ for any projector
$\Pi$ is a trace non-increasing CPM. We define the \emph{$\eps$-ball} of
states close to $\rho \in \subnormstates{\h}$ as
\begin{align*}
 \epsball{\rho} := \{ \rhot \in \subnormstates{\h} :
 \dpure{\rho}{\rhot} \leq \eps \} \, .
\end{align*}

We will now define the smooth min-entropy~\cite{renner05}.
\begin{definition}
  \label{def:min}
  Let $\eps \geq 0$ and $\rhoAB \in \subnormstates{\hAB}$. The
  \emph{min-entropy of A conditioned on B} is given by
  \begin{align*}
    \chmin{A}{B}{\rho} &:= \!\!\!\max_{\sigmaB \in \normstates{\hB}} \!\!\!
    \sup\, \{ \lambda \in \mathbb{R} : \rhoAB \leq 2^{-\lambda} \idA
    \kron \sigmaB \} \, .
  \end{align*}
  Furthermore, the \emph{smooth min-entropy of A conditioned on B} is defined as
  \begin{align*}
    \chmineps{A}{B}{\rho} &:= \max_{\rhotAB \in
      \epsball{\rhoAB}} \chmin{A}{B}{\rhot} \, .
  \end{align*}
\end{definition}
The conditional min-entropy is a measure of the uncertainty about the
state of a system A given quantum side information B. In particular,
if the system A describes a classical random variable (i.e.\ if the
state is CQ), the min-entropy can be interpreted as a guessing
probability.\footnote{See discussion in Section~\ref{sec:intro}
  and~\cite{koenig08} for details.}
For general quantum states, the smooth min-entropy satisfies
data-processing inequalities. For example, if a CPM is applied to the
B system or if a measurement is conducted on the A system, the smooth
min-entropy of A given B is guaranteed not to
decrease.\footnote{See~\cite{tomamichel09} for precise statements and
  proofs.}

Finally, we will need a fully quantum generalization of the collision
entropy (R\'enyi-entropy of order $2$).
\begin{definition}
  \label{def:coll}
  Let $\rhoAB \in \subnormstates{\hAB}$ and $\sigmaB \in
  \posops{\hB}$, then the \emph{collision entropy of A conditioned on B}
  of a state $\rhoAB$ given $\sigmaB$ is $-\log \idx{\Gamma}{C}(\rhoAB | \sigmaB)$, where
  \begin{align*}
    \idx{\Gamma}{C}(\rhoAB | \sigmaB) &:= \tr\, \big( \rhoAB
    (\idA \kron \sigmaB^{-\nicefrac{1}{2}}) \big)^2 \, .
  \end{align*}
\end{definition}
We will use the fact that the collision entropy provides an upper
bound on the min-entropy. The proof of the following statement can be
found in Appendix~\ref{app:coll} and constitutes one of the main
technical contributions of this work.

\begin{lemma}
\label{lemma:coll-bound}
Let $\rhoXB \in \subnormstates{\hi{XB}}$ be a CQ-state and $\bar{\eps} >
0$. Then, there exists a state $\sigmaB \in \normstates{\hB}$ such that
\begin{align}
  \label{eqn:min-coll-bound}
  \idx{\Gamma}{C}(\rhoXB | \sigmaB) \leq 2^{-
  \chmin{X}{B}{\rho} } \, .
\end{align}
Moreover, there exists a normalized CQ-state $\rhobXB \in \eepsball{
  \bar{\eps}}{\rhoXB}$ such that
\begin{align}
  \label{eqn:minalt-coll-bound}
  \idx{\Gamma}{C}(\rhobXB | \rhobB) \leq 
  2^{-\chmin{X}{B}{\rho} + \log (\frac{2}{\bar{\eps}^2} + 1)} \, .
\end{align}
\end{lemma}

\section{Proof of the Leftover Hash Lemma}
\label{sec:proof}

In this section we give bounds on the distance from uniform of the
quantum state after privacy amplification with two-universal and
$\delta$-almost two-universal hashing
(Theorems~\ref{thm:two-universal-hashing}
and~\ref{thm:almost-two-universal-hashing}). The proof
of Lemma~\ref{lemma:quant} then follows.

First, we extend the definition of the distance from uniform to
sub-normalized states for technical reasons.\footnote{Note that
  sub-normalized states have to be considered due to our
  definition of the smoothing of the min-entropy.}
\begin{definition}
  \label{def:dist-from-uniform}
  Let $\rhoAB \in \subnormstates{\hAE}$, then we define the \emph{distance
    from uniform of A conditioned on B} as
  \begin{align}
    \label{eqn:dist-from-uniform}
    \duni{A}{B}{\rho} := \min_{\sigmaB}\
    \frac{1}{2} \norm{\rhoAB - \omegaA \kron \sigmaB}{1} \, ,
  \end{align}
  where $\omegaA := \idA/\dim \hA$ and the minimum is taken over all $\sigmaB \in
  \posops{\hB}$ satisfying $\tr\,\sigmaB = \tr\,\rhoB$.
\end{definition}

As a first step, we bound the distance from uniform in terms of the
collision entropy. 
\begin{lemma}
  \label{lemma:uniform-to-collision}
  Let $\rhoAB \in \subnormstates{\hAB}$ and $\tauB \in
  \subnormstates{\hB}$ with $\supp{\tauB} \supseteq \supp{\rhoB}$, then
  \begin{align*}
   \duni{A}{B}{\rho} \leq \frac{1}{2}
    \sqrt{\idx{d}{A} \idx{\Gamma}{C}(\rhoAB | \tauB)  -
     \tr\big( \rhoB \tauB^{\!-\nicefrac{1}{2}}
      \!\rhoB \tauB^{\!-\nicefrac{1}{2}} \big) } \, .
  \end{align*}
\end{lemma}
\begin{proof}
  We apply the H\"older inequality (Lemma~\ref{lemma:hoelder} in
  Appendix~\ref{app:tech}) with parameters $r = t = 4$, $s = 2$, $A = C = \idA
  \kron \tauB^{\nicefrac{1}{4}}$ and $B = ( \idA \kron
  \tauB^{-\nicefrac{1}{4}} ) ( \rhoAB - \omegaA \kron \rhoB )(
  \idA \kron \tauB^{-\nicefrac{1}{4}} )$. This leads to
  \begin{align*}
    2\, \duni{A}{B}{\rho} &\leq \norm{ \rhoAB - \omegaA \kron \rhoB }{1}
    \\
    &= \norm{ABC}{1} \leq \norm{A^4}{1}^{\, \nicefrac{1}{4}}
    \norm{B^2}{1}^{\nicefrac{1}{2}} \norm{C^4}{1}^{\, \nicefrac{1}{4}}\\
    &\leq \sqrt{d_A\, \tr \big( (\rhoAB - \omegaA \kron
      \rhoB) (\idA \kron \tauB^{-\nicefrac{1}{2}}) \big)^2 } \, .
  \end{align*}
  We simplify the expression on the r.h.s.\ further using
  \begin{align*}
    &\tr \big( (\rhoAB - \omegaA \kron
    \rhoB) (\idA \kron \tauB^{-\nicefrac{1}{2}}) \big)^2 \\
    &\quad =\, \tr \big( \rhoAB (\idA \kron \tauB^{-\nicefrac{1}{2}})
    \big)^2 + \tr \big( (\omegaA \kron \rhoB) (\idA \kron \tauB^{-\nicefrac{1}{2}})
    \big)^2 \\
    &\qquad - 2 \tr \big( \rhoAB (\idA \kron
    \tauE^{-\nicefrac{1}{2}}) (\omegaA \kron \rhoB) (\idA \kron
    \tauE^{-\nicefrac{1}{2}}) \big) \\
    & \quad =\, \idx{\Gamma}{C}(\rhoAB | \tauB) -
    \frac{1}{\idx{d}{A}} \tr\big( \rhoB \tauB^{-\nicefrac{1}{2}}
    \rhoB \tauB^{-\nicefrac{1}{2}} \big) \, ,
  \end{align*}
  which concludes the proof.
\end{proof}
The above bound can be simplified by setting $\tauB = \rhoB$:
\begin{align}
  \label{eqn:uniform-bound-2}
  \duni{A}{B}{\rho} \leq \frac{1}{2} \sqrt{ \idx{d}{A}
    \idx{\Gamma}{C}(\rhoAB | \rhoB) - \tr\,\rhoB } \, .
\end{align}

We now consider a scenario where $X$ is picked from a set $\cX$ and
$E$ is a quantum system whose state may depend on $X$. The situation
is described by a CQ-state of the form
\begin{align}
  \label{eqn:cq-state}
  \rhoXE = \sum_x \proji{x}{x}{X} \kron \rhoE^{[x]} \, ,
\end{align}
where the probability of $x$ occurring is the trace of the
sub-normalized state $\rhoE^{[x]}$ and $\rhoE = \sum_x
\rhoE^{[x]}$. After applying a function $f: \cX \to \{0,1\}^{\times
  \ell}$ chosen at random from a family of hash functions $\cF$, the
resulting CQ-state is given by
\begin{align}
  \label{eqn:after-extraction-state}
  \idx{\rho}{FZE} = \sum_f \sum_z p_f \proji{f}{f}{F} \kron
  \proji{z}{z}{Z} \kron \rhoE^{[f,z]} \, ,
\end{align}
where $z \in \{0, 1\}^{\times \ell}$, $p_f = 1/\abs{\mathcal{F}}$ and
\begin{align}
  \label{eqn:rhoE-conversion}
  \rhoE^{[f,z]} := \sum_{x, f(x) = z} \rhoE^{[x]} \, .
\end{align}
% Note that $\idx{\rho}{FE} := \ptr{Z}(\idx{\rho}{FZE}) =
% \idx{\rho}{F} \kron \rhoE$ has product form and $\idx{\rho}{F} =
% \idi{F}/\abs{\mathcal{F}}$.
Formally, randomness extraction can be
modelled as a trace-preserving CPM,
$\mathcal{A}$, from $\hi{FX} \to \hi{FZ}$ that maps $\idx{\rho}{F}
\kron \idx{\rho}{XE} \mapsto (\mathcal{A} \kron \opidE)
(\idx{\rho}{F} \kron \idx{\rho}{XE}) = \idx{\rho}{FZE}$.

The followoing lemma yields a bound on the collision entropy of the
output of the hash function in terms of the collision entropy of the
input.

\begin{lemma}
  \label{lemma:privacy-amp}
  Let $\mathcal{F}$ be $\delta$-almost
  two-universal, let $\rhoXE$ and $\idx{\rho}{FZE}$ be defined as
  in~\eqref{eqn:cq-state} and~\eqref{eqn:after-extraction-state},
  respectively, and let $\tauE \in \normstates{\hE}$. Then, 
  \begin{align*}
    \idx{\Gamma}{C}(\idx{\rho}{FZE} | \idx{\rho}{F}
    \kron \tauE) \leq \idx{\Gamma}{C}(\rhoXE|\tauE) + \delta\,
    \tr\,(\rhoE \tauE^{\!-\nicefrac{1}{2}}
    \rhoE \tauE^{\!-\nicefrac{1}{2}}) \, .
  \end{align*}
\end{lemma}
\begin{proof}
  The collision entropy on the l.h.s.\ can be rewritten as an expectation value
  over $F$, that is
  \begin{align*}
    &\idx{\Gamma}{C}(\idx{\rho}{FZE} | \idx{\rho}{F}
    \kron \tauE) \\
    &\quad =\, \tr\Big(\idx{\rho}{FZE} (p_f \idi{FZ} \kron
    \tauE)^{-\nicefrac{1}{2}} \idx{\rho}{ZEF} (p_f \idi{FZ}
    \kron \tauE)^{-\nicefrac{1}{2}} \Big) \\
    &\quad =\, {\sum}_f\, p_f\! \sum_z \tr\big(
    \proji{f}{f}{F} \!\kron\! \proji{z}{z}{Z}\! \kron
    \rhoE^{[f,z]} \tauE^{\!\!-\nicefrac{1}{2}}\! \rhoE^{[f,z]}
    \tauE^{\!\!-\nicefrac{1}{2}} \big) \\
    &\quad =\, \mathop{\mathbb{E}}_{F \in \mathcal{F}} \Big[ \sum_z
    \tr\,( \rhoE^{[F,z]} \tauE^{\!-\nicefrac{1}{2}} \rhoE^{[F,z]}
    \tauE^{\!-\nicefrac{1}{2}} ) \Big] \\
    &\quad =\, \sum_{x, x'} \mathop{\mathbb{E}}_{F \in \mathcal{F}} \Big[
    \sum_z \delta_{F(x) = z} \delta_{F(x') = z} \Big]
    \tr\,( \rhoE^{[x]} \tauE^{\!-\nicefrac{1}{2}} \rhoE^{[x']}
    \tauE^{\!-\nicefrac{1}{2}} ) \, .
  \end{align*}
  We have used~\eqref{eqn:rhoE-conversion} to substitute for
  $\rhoE^{[F,z]}$ in the last step.
  The expectation value can be evaluated using the defining
  property~\eqref{eqn:almost-two-universal} of
  $\delta$-almost two-universal families. We get
  \begin{align*}
    \mathop{\mathbb{E}}_{F \in \mathcal{F}} \Big[
    \sum_z \delta_{F(x) = z} \delta_{F(x') = z} \Big] \leq \delta
  \end{align*}
  if $x \neq x'$ and $1$ otherwise. We use this relation and the fact that the
  trace terms are positive to bound
  \begin{align*}
    &\idx{\Gamma}{C}(\idx{\rho}{FZE} | \idx{\rho}{F}
    \kron \tauE) \\
    &\quad \leq \sum_x \tr( \rhoE^{[x]} \tauE^{\!\!-\nicefrac{1}{2}}
    \rhoE^{[x]} \tauE^{\!\!-\nicefrac{1}{2}} ) + \delta \!\sum_{x \neq
      x'}\! \tr( \rhoE^{[x]} \tauE^{\!\!-\nicefrac{1}{2}}
    \rhoE^{[x']} \tauE^{\!\!-\nicefrac{1}{2}} )  \, . %\\
    % &\quad =\, \idx{\Gamma}{C}(\idx{\rho}{XE} | \rhoE) + \delta\,
    % \tr( \rhoE \tauE^{-\nicefrac{1}{2}}
    % \rhoE \tauE^{-\nicefrac{1}{2}} )
  \end{align*}
  We now complete the second sum with the terms where $x = x'$ to get
  the statement of the lemma.
\end{proof}

If we set $\tauE = \rhoE$, the result can be simplified further:
\begin{align}
  \label{eqn:privacy-amp-1}
  \idx{\Gamma}{C}(\idx{\rho}{FZE} | \idx{\rho}{F}
  \kron \rhoE) \leq \idx{\Gamma}{C}(\rhoXE|\rhoE) + \delta\,
  \tr\,\rhoE \, .
\end{align}

We are now ready to give a bound on the distance from uniform
$\duni{Z}{FE}{}$ after privacy amplification with two-universal and
$\delta$-almost two-universal families of hash functions. Note that
we consider the distance from uniform conditioned on F as well as E.
This describes the situation where the chosen hash function (the
value $f$) is published after its use
(strong extractor regime).

The distance from uniform conditioned on E averaged over
the choice of $f$ is given by
\begin{align*}
  \sum_{f} p_f \duni{Z}{E}{\rho^{[f]}} \, , \quad
  \textrm{where}\ \ \idx{\rho}{ZE}^{[f]} := \sum_{z} \proji{z}{z}{Z} \kron
  \rhoE^{[f,z]} \, 
\end{align*}
and it can be bounded in terms of $\duni{Z}{FE}{}$ as
\begin{align}
  \sum_{f} p_f\, \duni{Z}{E}{\rho^{[f]}} &\leq \frac{1}{2} \sum_f p_f\,
  \normbig{ \idx{\rho}{ZE}^{[f]} - \omegaZ \kron \sigmaE }{1}
  \nonumber\\
  % &= \frac{1}{2} \normbig{ \idx{\rho}{FZE} - p_f \idi{F} \kron \omegaZ
  %   \kron \sigmaE}{1}
  \label{eqn:average-bound}
  &= \duni{Z}{EF}{\rho} \, ,
\end{align}
where $\sigmaE$ optimizes~\eqref{eqn:dist-from-uniform} for
$\duni{Z}{EF}{\rho}$. Hence, an upper bound on $\duni{Z}{FE}{}$ implies
an upper bound on the average distance to uniform conditioned on E as well.

For two-universal hashing, we get the following bound (see
also~\cite{renner05}).
\begin{theorem}
  \label{thm:two-universal-hashing}
  Let $\mathcal{F}$ be two-universal and let $\rhoXE$ and
  $\idx{\rho}{ZEF}$ be defined as in~\eqref{eqn:cq-state}
  and~\eqref{eqn:after-extraction-state}, respectively. Then, for any
  $\eps \geq 0$,
  \begin{align*}
    \duni{Z}{FE}{\rho} \leq \eps + \frac{1}{2}
    \sqrt{2^{\ell - \chmineps{X}{E}{\rho}}} \, .
  \end{align*}
\end{theorem}

\begin{proof}
  We use Lemma~\ref{lemma:uniform-to-collision} to bound
  $\duni{Z}{FE}{\rho}$.  In particular, we set $\idx{\tau}{FE} :=
  \idx{\rho}{F} \kron \idx{\tau}{E}$ to get
  \begin{align*}
    2\duni{Z}{FE}{\rho} &\leq \sqrt{2^\ell \idx{\Gamma}{C}
      (\idx{\rho}{ZFE} | \idx{\tau}{FE}) - \tr\,
      (\rhoE \tauE^{\!-\nicefrac{1}{2}} \rhoE \tauE^{\!-\nicefrac{1}{2}})} \\
    &\leq \sqrt{2^\ell \idx{\Gamma}{C}(\idx{\rho}{XE}|\tauE)} \, ,
  \end{align*}
  where we have used Lemma~\ref{lemma:privacy-amp} and that
  $\mathcal{F}$ is two-universal ($\delta \leq 2^{-\ell}$) in the last
  step. The r.h.s.\ can be expressed in terms of a min-entropy
  using~\eqref{eqn:min-coll-bound}. With an appropriate choice of
  $\tauE$, we have
  \begin{align}
    \label{eqn:non-smooth-key-bound}
    2\duni{Z}{FE}{\rho} \leq
    \sqrt{2^{\ell - \chmin{X}{E}{\rho}}} \, .
  \end{align}
  We have now shown the statement of the theorem for the case $\eps =
  0$.

  Finally, the bound can be expressed in terms of a smooth
  min-entropy. Let $\idx{\rhot}{XE} \in \epsball{\rhoXE}$ be the
  CQ-state (cf.~Lemma~\ref{lemma:cq-smoothing}) that optimizes the
  smooth min-entropy $\chmineps{X}{E}{\rho} = \chmin{X}{E}{\rhot}$. We
  define $\idx{\rhot}{FZE} := (\mathcal{A} \kron \opidE)
  (\idx{\rho}{F} \kron \idx{\rhot}{XE})$ and note that privacy
  amplification can only decrease the purified
  distance~\eqref{eqn:purified-monotone}, i.e.
  \begin{align*}
    \frac{1}{2} \norm{\idx{\rho}{FZE} - \idx{\rhot}{FZE}}{1} \leq
    \dpure{\idx{\rho}{FZE}}{\idx{\rhot}{FZE}} \leq
    \dpure{\rhoXE}{\idx{\rhot}{XE}} \leq \eps \, .
  \end{align*}
  Moreover, let $\idx{\sigmat}{FE}$ be the state that minimizes the
  distance from uniform $d_u(Z|FE)_{\rhot}$. Then, 
  \begin{align*}
    2 \duni{Z}{FE}{\rho} &\leq \norm{\idx{\rho}{FZE} - \idx{\omega}{Z}
      \kron \idx{\sigmat}{FE}}{1} \\
    &\leq \norm{\idx{\rho}{FZE} -  \idx{\rhot}{FZE}}{1} +
    \norm{\idx{\rhot}{FZE} - \idx{\omega}{Z} \kron \idx{\sigmat}{FE}}{1} \\
    & \leq 2 \eps + 2 \duni{Z}{FE}{\rhot} \, .
  \end{align*}
  We now apply~\eqref{eqn:non-smooth-key-bound} for
  $\idx{\rhot}{FZE}$ (instead of $\idx{\rho}{FZE}$) to get
  \begin{align*}
    \duni{Z}{FE}{\rho} &\leq \eps + \frac{1}{2} \sqrt{2^{\ell -
        \chmin{X}{E}{\rhot}}} \\
    &= \eps + \frac{1}{2} \sqrt{2^{\ell - \chmineps{X}{E}{\rho}}}\, ,
  \end{align*}
  which concludes the proof.
\end{proof}

Next, we consider the case of $\delta$-almost two-universal hashing.
\begin{theorem}
  \label{thm:almost-two-universal-hashing}
  Let $\mathcal{F}$ be $\delta$-almost two-universal and let $\rhoXE$
  and $\idx{\rho}{ZEF}$ be defined as in~\eqref{eqn:cq-state}
  and~\eqref{eqn:after-extraction-state}, respectively. Then, for any
  $\eps \geq 0$ and $\bar{\eps} > 0$,
  \begin{align*}
    \duni{Z}{FE}{\rho} \leq \eps \!+\! \bar{\eps} \!+\!
    \frac{1}{2}\sqrt{(2^\ell \delta \!-\! 1) \! +\! 2^{\ell - \chmineps{X}{B}{\rho} +
        \log (\frac{2}{\bar{\eps}^2} + 1) } } .
  \end{align*}
\end{theorem}
\begin{proof}
  We use Lemma~\ref{lemma:uniform-to-collision} as
  in~\eqref{eqn:uniform-bound-2} to bound $\duni{Z}{FE}{\rho}$.
  For normalized $\idx{\rho}{ZFE}$, we find
  \begin{align*}
    2 \duni{Z}{FE}{\rho} &\leq \sqrt{ 2^\ell \idx{\Gamma}{C}(\idx{\rho}{FZE} |
      \idx{\rho}{F} \kron \rhoE) - 1 } \\
    &\leq \sqrt{2^\ell \idx{\Gamma}{C}(\rhoXE|\rhoE) + (2^\ell
      \delta - 1)} \, ,
  \end{align*}
  where we used Lemma~\ref{lemma:privacy-amp} as
  stated in~\eqref{eqn:privacy-amp-1}.

  The smoothing of the above equation is achieved using the same arguments
  as in the proof of Theorem~\ref{thm:two-universal}.
  However, this time we need to include an additional smoothing
  parameter $\bar{\eps} > 0$ in order to be able to
  apply~\eqref{eqn:minalt-coll-bound}.
 
  Let $\idx{\rhot}{XE} \in \epsball{\rhoXE}$ be the CQ-state
  (cf.~Lemma~\ref{lemma:cq-smoothing}) that optimizes the smooth
  min-entropy $\chmineps{X}{E}{\rho} = \chmin{X}{E}{\rhot}$ and let
  $\idx{\rhob}{XE} \in \eepsball{\bar{\eps}}{\idx{\rhot}{XE}}$ be the
  CQ-state (cf.~Lemma~\ref{lemma:coll-bound}) that satisfies
  \begin{align}
    \idx{\Gamma}{C}(\idx{\rhob}{XE} | \rhobE) &\leq 
    2^{-\chmin{X}{E}{\rhot} + \log (\frac{2}{\bar{\eps}^2} + 1)} \nonumber\\
    \label{eqn:almost-coll-bound}
    &= 2^{-\chmineps{X}{E}{\rho} + \log (\frac{2}{\bar{\eps}^2} + 1)} \, .
  \end{align}
  Then, $\idx{\rhob}{XE} \in \eepsball{\eps + \bar{\eps}}{\rhoXE}$
  holds due to the triangle inequality of the purified distance.
  Moreover, we define the state after randomness extraction,
  $\idx{\rhob}{FZE} := (\mathcal{A} \kron \opidi{E})(\idx{\rho}{F}
  \kron \idx{\rhob}{XE})$. Following the arguments laid out in the
  proof of Theorem~\ref{thm:two-universal-hashing}, we have
  \begin{align*}
    \duni{Z}{FE}{\rho} &\leq \eps + \bar{\eps} +  \duni{Z}{FE}{\rhob} \\
    &\leq \eps + \bar{\eps} + \frac{1}{2}\sqrt{2^{\ell}
      \idx{\Gamma}{C}(\idx{\rhob}{XE} | \rhobE) + (2^\ell \delta - 1)} \, .
  \end{align*}
  This can be bounded using~\eqref{eqn:almost-coll-bound}, which
  concludes the proof.
\end{proof}

The proof of the Leftover Hash Lemma stated in the introduction
(Lemma~\ref{lemma:quant}) follows when we set $\eps = 0$ in
Theorem~\ref{thm:two-universal} and
Theorem~\ref{thm:almost-two-universal}. To see this, note that the
statements of two theorems can be expressed in terms of the
distance from uniform averaged over the choice of $f$
using~\eqref{eqn:average-bound}.

\section{Explicit Constructions with Shorter Seeds}
\label{sec:constructions}

Here, we combine known constructions of two-universal and
$\delta$-almost two-universal hash functions and discuss their use for
randomness extraction with shorter random seeds.  We consider a
scenario where $X$ is an $n$-bit string $x \in \{0, 1\}^{\times n}$
and $E$ is a quantum system. The challenge is typically to optimize
the following parameters:
\begin{enumerate}
\item[a)] the error described by the distance from uniform, $e :=
  \duni{Z}{FE}{}$, which should be small,
\item[b)] the length of the extracted key, $\ell$, which one wants to
  make as large as possible (close to $\chmineps{X}{E}{}$) and
\item[c)] the length of the random seed, $s := \log
  \abs{\mathcal{F}}$, needed to choose $f$, which one wants to keep
  small.
\end{enumerate}
The latter point is important in practical implementations of privacy
amplification, for example in quantum key distribution (QKD), where
the choice of $f$ has to be communicated between two parties.

We will first review the explicit constructions of ($\delta$-almost)
two-universal hash functions used in this section. In~\cite{carter79}, 
Carter and Wegman proposed several constructions of two-universal
function families, trying to minimize the size of $\cF$. An example of
a two-universal set of hash functions with $|\cF| = 2^n$ is the set
$\cF = \{f_\alpha\}_{\alpha \in \{0,1\}^n}$ consisting of elements
\begin{align} \label{eq:multconstruction}
  \begin{array}{ccccc}
    f_{\alpha} & : & \{0,1\}^n & \longrightarrow & \{0,1\}^\ell \\
         &  & x & \longmapsto &  x \cdot \alpha \mod 2^\ell
  \end{array}
\end{align}
where $x \cdot \alpha$ denotes the multiplication in the field
$\GF(2^n)$.
%and where $[\cdot]_\ell$ is the mapping that extracts the $\ell$
%least significant bits.
The fact that $\cF$ is two-universal can be
readily verified by considering the difference $f_{\alpha}(x) -
f_{\alpha}(x') = (x - x') \cdot \alpha \mod 2^\ell$ and noting that the mapping
$\alpha \mapsto (x-x') \cdot \alpha$ is a bijection if $x - x' \neq
0$.

With $\delta$-almost two-universal families, a larger value of $\delta$
typically allows for a smaller set $\cF$. This is nicely illustrated
by the following well-known construction based on polynomials. Let
$\mathbb{F}$ be an arbitrary field and let $r$ be a positive
integer. We define the family $\mathcal{F} = \{f_\alpha\}_{\alpha \in
  \mathbb{F}}$ of functions
\begin{align} \label{eq:polyconstruction}
  \begin{array}{ccccc}
    f_{\alpha} & : & \mathbb{F}^r& \longrightarrow & \mathbb{F} \\
         &  &  (x_1, \ldots, x_r)  & \longmapsto & \sum_{i=1}^r x_i \alpha^{r-i}    \ .
  \end{array}
\end{align}
Using the fact that a polynomial of degree $r-1$ can only have $r-1$
zeros, it is easy to verify that $\mathcal{F}$ is $\delta$-almost
two-universal, for $\delta = (r-1)/|\mathbb{F}|$.

Another method to construct $\delta$-almost two-universal families of
hash functions is to concatenate two such families. We will use the
following lemma by Stinson (see Theorem~5.4 in~\cite{stinson94}).
\begin{lemma}
  \label{lemma:concatenation}
  Let $\mathcal{F}_1$ be $\delta_1$-almost two-universal from
  $\{0,1\}^{\times n}$ to $\{0,1\}^{\times k}$ and let
  $\mathcal{F}_2$ be $\delta_2$-almost two-universal from $\{0,1\}^{\times
    k}$ to $\{0,1\}^{\times \ell}$. Then, the family
  $\mathcal{G} := \left\{ f_2 \circ f_1 : f_1 \in
    \mathcal{F}_1, f_2 \in \mathcal{F}_2 \right\}$ consisting of all
  concatenated hash functions is $(\delta_1 + \delta_2)$-almost
  two-universal.
\end{lemma}

Combining the general results on $\delta$-almost two-universal
hashing of Section~\ref{sec:proof} with the explicit
constructions described above, we obtain the following statements.

If we do not care about $s$, we may choose a two-universal family of
hash functions and recover a result by Renner~\cite{renner05}:
\begin{theorem}
  \label{thm:two-universal}
  There exists a family of hash functions from $\{0, 1\}^{\times n}$ to $\{0,
  1\}^{\times \ell}$ satisfying
  \begin{align*}
    s = n \;\; \textrm{and} \;\; e \leq \eps + \frac{1}{2}
    \sqrt{2^{\ell - \chmineps{X}{E}{\rho}}} \quad \textrm{for any}\
    \eps \geq 0 .
  \end{align*}
\end{theorem}

\begin{proof}
  We apply Theorem~\ref{thm:two-universal-hashing} using the
  two-universal family constructed in~\eqref{eq:multconstruction}, which yields $s = \log |\mathcal{F}| = n$.
\end{proof}

We now show that we can choose a family of hash functions such that
$s$ is proportional to the key length $\ell$ instead of the
input string length $n$.
\begin{theorem}
  \label{thm:almost-two-universal}
  There exists a family of hash functions from $\{0, 1\}^{\times n}$
  to $\{0, 1\}^{\times \ell}$ satisfying
  \begin{align*}
    s &= 2 \lfloor \ell + \log(n/\ell)
+ \log(1/\eps^2) - 1 \rfloor  \quad \textrm{and}\\
    e &\leq 3\eps + \frac{1}{2}
    \sqrt{2^{\ell - \chmineps{X}{E}{\rho}  +
        \log (\frac{2}{\eps^2} + 1) }} \quad \textrm{for any}\
    \eps > 0 .
  \end{align*}
\end{theorem}
\begin{proof}
  We use the standard classical way of concatenating two hash
  functions to obtain the required
  parameters~\cite{srinivasan99}. For the first function, we set
  $k = \lfloor \ell + \log(n/\ell) + \log(1/\eps^2) \rfloor$ and use the field
  $\mathbb{F} = \GF(2^k)$ in the polynomial-based hash construction
  from~\eqref{eq:polyconstruction}. Interpreting the $n$-bit strings
  as $r=\lceil n/k \rceil$ blocks of $k$ bits, the first hash function
  maps from $\{0,1\}^{\times n}$ to $\{0,1\}^{\times k}$ and requires
  a $k$-bit seed. Then, regular two-universal hashing
  from~\eqref{eq:multconstruction} with a seed length of again $k$
  bits is used to map from $\{0,1\}^{\times k}$ to $\{0,1\}^{\times
    \ell}$.  The two seed lengths add up to
  $s = 2k = 2 \lfloor \ell + \log(n/\ell) + \log(1/\eps^2) \rfloor$.

  Polynomial-based hashing achieves a $\delta_1$ of at most
  \begin{align*}
    \frac{r-1}{2^k} \leq \frac{n}{k\, 2^k} \leq \frac{4\,\ell\,
      \eps^2}{k\, 2^\ell} \leq \frac{4 \eps^2}{2^{\ell}}
  \end{align*}
  by the choice of $r$ and the fact that $k \geq \ell + \log(n/\ell) +
  \log(1/\eps^2) - 2$.  Together with the $\delta_2 \leq 2^{-\ell}$
  from the two-universal hashing, we get from
  Lemma~\ref{lemma:concatenation} that this construction yields a
  $\delta_1+\delta_2 \leq \frac{1 + 4\eps^2}{2^\ell}$-almost
  two-universal family of hash functions.  Inserting this expression
  for $\delta$ into Theorem~\ref{thm:almost-two-universal-hashing} and
  setting $\bar{\eps} = \eps$ yields
  \begin{align*}
    e \leq 2\eps + \frac{1}{2} \sqrt{2^{\ell - \chmineps{X}{E}{\rho} +
        \log(\frac{2}{\eps^2} + 1)} + 4\eps^2} \, .
  \end{align*}
  The theorem then follows as an upper bound to this expression.
\end{proof}

\section*{Acknowledgment}

We thank Roger Colbeck and Johan {\AA}berg for useful discussions and
comments. MT and RR acknowledge support from the Swiss National
Science Foundation (grant No. 200021-119868). CS is supported by a NWO
VICI project.

\appendix

\section{Technical Results}
\label{app:tech}

The first lemma is an application of Uhlmann's
theorem~\cite{uhlmann85} to the purified distance\footnote{The main
advantage of the purified distance over the trace distance is that
we can always find extensions and purifications without increasing
the distance.}  (see~\cite{tomamichel09} for a proof).
\begin{lemma}
\label{lemma:uhlmann-pure}
Let $\rho, \tau \in \subnormstates{\h}$, $\h' \iso \h$ and $\varphi
\in \h \kron \h'$ be a purification of $\rho$. Then, there exists a
purification $\vartheta \in \h \kron \h'$ of $\tau$ with $P(\rho,
\tau) = P(\varphi, \vartheta)$.
\end{lemma}
\begin{corollary}
\label{cor:uhlmann-ext}
Let $\rho, \tau \in \subnormstates{\h}$ and $\rhob \in
\subnormstates{\h \kron \h'}$ be an extension of $\rho$. Then, there
exists an extension $\taub \in \subnormstates{\h \kron \h'}$ of $\tau$
with $P(\rho, \tau) = P(\rhob, \taub)$.
\end{corollary}
In the following, we apply this result to an $\eps$-ball of pure states,
$\epsballpure{\rho} := \{ \rhot \in \epsball{\rho} :
\rank\,\rhot = 1 \}$.

\begin{corollary}
\label{cor:pure-ball}
Let $\rho \in \subnormstates{\h}$ and $\varphi \in \h
  \kron \h'$ be a purification of $\rho$. Then,
\begin{equation*}
  \epsball{\rho} \supseteq \{ \rhot \in \subnormstates{\h} : \exists\, \phit \in \epsballpure{\varphi}
  \ \textrm{s.t.}\ \rhot = \ptr{\h'}{\phit} \}
\end{equation*}
and equality holds if the Hilbert space dimensions satisfy $\dim \h'
\geq \dim \h$.
\end{corollary}

The following lemma establishes a fundamental property of pure
bipartite states, namely that every linear operator applied to
one subsystem has a dual on the other subsystem, such that the
resulting pure state is the same. 
\begin{lemma}
\label{lemma:mirror}
Let $\phiAB \in \posops{\hAB}$ be pure, $\rhoA = \ptr{B}\,
\phiAB$, $\rhoB = \ptr{A}\,\phiAB$ and let $X \in
\linops{\hA}$ be an operator with support and image in $\supp{\rhoA}$. Then,
\begin{align*}
  \big(X \kron \idB \big) \keti{\phi}{AB} = \big( \idA \kron (
  \rhoB^{\nicefrac{1}{2}} X^T \rhoB^{-\nicefrac{1}{2}} ) \big)
  \keti{\phi}{AB} \, ,
\end{align*}
where the transpose is taken with regard to the Schmidt basis of
$\phiAB$.
\end{lemma}  
\begin{proof}
We introduce the Schmidt decomposition $\keti{\phi}{AB} = \sum_i
\sqrt{\lambda_i}\, \keti{i}{A} \kron \keti{i}{B}$. Clearly,
$\big( \idA \kron
\rhoB^{-\nicefrac{1}{2}} \big) \keti{\phi}{AB} = \sum_i\,
\keti{i}{A} \kron \keti{i}{B} =: \keti{\gamma}{AB}$ is the
(unnormalized) fully entangled state on the support of $\rhoA$ and
$\rhoB$. It is easy to verify that $(X \kron \idB) \keti{\gamma}{AB}
= (\idA \kron X^T) \keti{\gamma}{AB}$, where the transposed matrix
is given by $X^T = \sum_{i,j} \, \bracketi{i}{X}{j}{A} \,
\proji{j}{i}{B}$.
\end{proof}
\begin{corollary}
\label{cor:mirror-function}
Let $\phiAB \in \posops{\hAB}$ be pure, $\rhoA = \ptr{B}\,\phiAB$,
$\rhoB = \ptr{A}\,\phiAB$ and $f: \mathbb{R}^+ \to \mathbb{R}$ a
real-valued function, then
\begin{align*}
  \big( f(\rhoA) \kron \idB \big) \keti{\phi}{AB} = \big( \idA \kron
  f(\rhoB) \big) \keti{\phi}{AB} \, .
\end{align*}
\end{corollary}
We define the notion of a \emph{dual projector} with regard to a pure
state using the following corollary:
\begin{corollary}
\label{cor:mirror-project}
Let $\keti{\phi}{AB} \in \hAB$ be pure, $\rhoA = \ptr{B}\,\phiAB$,
$\rhoB = \ptr{A}\,\phiAB$ and let $\PiA \in \posops{\hA}$ be a
projector in $\supp{\rhoA}$. Then, there exists a dual projector
$\PiB$ on $\hB$ such that
\begin{align*}
  \big(\PiA \kron \rhoB^{-\nicefrac{1}{2}}\big) \keti{\phi}{AB} =
  \big(\rhoA^{-\nicefrac{1}{2}} \kron \PiB\big) \keti{\phi}{AB} \, .
\end{align*}
\end{corollary}
%
%\begin{proof}
%  The first statement follows from Lemma~\ref{lemma:mirror} after we
%  note that $f(\rhoB) = f(\rhoA)^T$ commutes with $\rhoB$ and the second
%  when substituting $X = \PiA \rhoA^{-\nicefrac{1}{2}}$.
%\end{proof}

The next Lemma gives a bound on the purified distance of a state
$\rho$ and a projected state $\Pi \rho \Pi$.
\begin{lemma}
\label{lemma:project-dist}
Let $\rho \in \subnormstates{\h}$ and $\Pi$ a projector on $\h$, then
\begin{align*}
  P \big(\rho, \Pi\rho\Pi \big) \leq \sqrt{2\, \trace{\Pi^\perp
      \rho } - \trace{\Pi^\perp\rho}^2} \, ,
\end{align*}
where $\Pi^\perp = \id - \Pi$ is the complement of $\Pi$ on $\h$.
\end{lemma}

\begin{proof}
 The generalized fidelity between the two states can be bounded using
 $\trace{\Pi\rho} \leq \trace{\rho}$. We have
\begin{align*}
  \bar{F}(\rho, \Pi\rho\Pi) &\geq \trace{\Pi \rho} + 1 - \tr\,\rho
  = 1 - \trace{\Pi^\perp\rho} \, .
\end{align*}
The desired bound on the purified distance follows from its
definition.
\end{proof}

We also need a H\"older inequality for linear operators and unitarily
invariant norms (see~\cite{bhatia97} for a proof).  Here, we state a
version for three operators and the trace norm:
\begin{lemma}
\label{lemma:hoelder}
Let $A$, $B$ and $C$ be linear operators and $r, s, t > 0$ such that
$\frac{1}{r} + \frac{1}{s} + \frac{1}{t} = 1$, then
\begin{align*}
  \norm{ABC}{1} \leq \norm{\abs{A}^r}{1}^{\ \frac{1}{r}}
  \norm{\abs{B}^s}{1}^{\ \frac{1}{s}} \norm{\abs{C}^t}{1}^{\
    \frac{1}{t}}  . 
\end{align*}
\end{lemma}

The following lemma makes clear that the min-entropy smoothing of a
state will not destroy its CQ structure.
\begin{lemma}
\label{lemma:cq-smoothing}
Let $\rhoXB$ be a CQ-state of the form $\rhoXB = \sum_x \proj{x}{x} \kron
\rhoB^{[x]}$. Then, the state $\idx{\rhot}{XB} \in \epsball{\rhoXB}$
that optimizes $\chmineps{X}{B}{\rho} = \chmin{X}{B}{\rhot}$ is of the same form.
\end{lemma}

\begin{proof}
Let $\rhotAB$ be any state in $\epsball{\rhoXB}$. We
can establish a CQ-state $\idx{\rhot}{XB}$ by measuring A in the basis
determined by $X$. This operation will not increase the distance
$P(\rhotAB, \rhoXB)$ (cf.~\cite{tomamichel09}, Lemma 7) and not
decrease the min-entropy (cf.~\cite{tomamichel09}, Theorem 19).
Thus, we can conclude that the optimal state is CQ.
\end{proof}

\section{Alternative Entropic Quantities}
\label{app:equivalent}
\label{app:alt}

Here, we discuss two alternative entropic quantities,
$\chminalteps{A}{B}{}$ and $\chmaxalteps{A}{B}{}$ and show that they
are equivalent (up to terms in $\log \eps$) to the smooth min-entropy
and smooth max-entropy, respectively. Some of the technical results of
this appendix will be used to give a bound on the collision entropy in
terms of the smooth min-entropy (cf.~Appendix~\ref{app:coll} and
Lemma~\ref{lemma:coll-bound}).

First, note that conditional entropies can be defined in terms of
relative entropies, as is well-known for the case of the von Neumann
entropy. Let $\rhoAB$ be a bipartite quantum state. Then, the
condtional von Neumann entropy of A given B is defined as
\begin{align}
 \chvn{A}{B}{\rho} :\!&= H(\rhoAB) - H(\rhoB) \nonumber\\
 &= - D(\rhoAB\,\|\,\idA \kron \rhoB) \label{eqn:cond-from-rel-1} \\
 &= - \min_{\sigmaB \in \normstates{\hB}}  D(\rhoAB\,\|\,\idA \kron
 \sigmaB) \label{eqn:cond-from-rel-2} \, ,
\end{align}
where we used Klein's inequality~\cite{klein31,nielsen00} in the last
step. The relative entropy is defined as $D(\rho\,\|\,\tau) :=
\trace{\rho (\log \rho - \log \tau)}$ and $H(\rho) := -\trace{\rho
 \log \rho}$.

We will now define the smooth min-entropy and an alternative to the
smooth entropy as first introduced in~\cite{renner05}.  The definition
of two versions of the min-entropy is parallel to the case of the von
Neumann entropy above; however, the two
identities~\eqref{eqn:cond-from-rel-1} and~\eqref{eqn:cond-from-rel-2}
now lead to different definitions. We follow~\cite{datta08} and first
introduce the \emph{max relative entropy}. For two positive operators
$\rho \in \subnormstates{\h}$ and $\tau \in \posops{\h}$ we define
\begin{align*}
  D_{\textrm{max}}(\rho\,\|\,\tau) := \inf \{ \lambda \in \mathbb{R} : \rho
  \leq 2^{\lambda} \tau \} \, .
\end{align*}
\begin{definition}
  \label{def:min}
  Let $\eps \geq 0$ and $\rhoAB \in \subnormstates{\hAB}$. The
  min-entropy and the \emph{alternative min-entropy} of A
  conditioned on B are given by
  \begin{align*}
    \chmin{A}{B}{\rho} &= \max_{\sigmaB \in \normstates{\hB}}
    - D_{\textrm{max}}(\rhoAB\,\|\,\idA \kron \sigmaB) \quad \textrm{and} \\
    \chminalt{A}{B}{\rho} &:= - D_{\textrm{max}}(\rhoAB\,\|\,\idA \kron \rhoB) \, ,
  \end{align*}
  respectively.
  Furthermore, the smooth min-entropy and the \emph{alternative
    smooth min-entropy} of A conditioned on B are defined as
  \begin{align*}
    \chmineps{A}{B}{\rho} &= \max_{\rhotAB \in
      \epsball{\rhoAB}} \chmin{A}{B}{\rhot} \, \quad
    \textrm{and} \\
    \chminalteps{A}{B}{\rho} &:= \max_{\rhotAB \in
      \epsball{\rhoAB}} \chminalt{A}{B}{\rhot} \, .
  \end{align*}
\end{definition}
The \emph{smooth max-entropies} can be defined as duals of the smooth
min-entropies.
\begin{definition}
  \label{def:max}
  Let $\eps \geq 0$ and $\rhoAB \in \subnormstates{\hAB}$, then
  we define the \emph{smooth max-entropy} and the \emph{alternative
    smooth max-entropy} of A conditioned on B as
  \begin{align*}
    \chmaxeps{A}{B}{\rho} &:= - \chmineps{A}{C}{\rho} \quad
    \textrm{and} \\
    \chmaxalteps{A}{B}{\rho} &:= -
    \chminalteps{A}{C}{\rho} \, ,
  \end{align*}
  where $\rhoABC \in \subnormstates{\hABC}$ is any purification of $\rhoAB$.
\end{definition}
The max-entropies are well-defined since the min-entropies are invariant
under local isometries on the C system (cf.~\cite{tomamichel09} and
Lemma~\ref{lemma:altiso}) and, thus, independent of the
chosen purification.
The non-smooth max-entropies $\chmax{A}{B}{\rho}$ and
$\chmaxalt{A}{B}{\rho}$ are defined as the limit $\eps \to 0$ of the
corresponding smooth quantities.
The alternative max-entropy is discussed in
Appendix~\ref{app:altmax}, where it is shown that (cf.\ also \cite{berta08})
\begin{align}
  \label{eqn:altmax}
  \chmaxalt{A}{B}{\rho} &= \!\!\max_{\sigmaB \in \normstates{\hB}}\!
  \log \Trace{ \Pi_{\rhoAB} (\idA \kron \sigmaB) } \, ,
\end{align}
where $\Pi_{\rhoAB}$ is the projector onto the support of $\rhoAB$.
Furthermore, we find that
\begin{align}
  \label{eqn:altmaxeps}
  \chmaxalteps{A}{B}{\rho} &= \!\! \inf_{\hi{B'} \supseteq
    \hB} \ \min_{\idx{\rhot}{AB'} \in
    \epsball{\idx{\rho}{AB'}}}\! \chmaxalt{A}{B'}{\rhot}\, ,
\end{align}
where the infimum is taken over all embeddings $\idx{\rho}{AB'}$ of
$\rhoAB$ into $\hA \kron \hi{B'}$. In fact, it is sufficient to consider an
embedding into a space of size $\dim \hi{B'} = \rank\,\{\rhoAB\}
\cdot \dim \hA$.

The first definition of the smooth max-entropy, $\chmaxeps{A}{B}{}$,
is used in~\cite{koenig08,tomamichel08} and is found to have many
interesting properties, e.g.\ it satisfies a data-processing
inequality~\cite{tomamichel09}.  The alternative definition,
$\chmaxalteps{A}{B}{}$, was first introduced in~\cite{renner05} and is
used to quantitatively characterize various information theoretic
tasks (cf.\ e.g.~\cite{datta08,mosonyidatta08,datta09}). Here, we find
that the two smooth min-entropies and the two smooth max-entropies are
pairwise equivalent up to terms in $\log \eps$. Namely, the following
lemma holds:
\begin{lemma}
  \label{lemma:equivalence}
  Let $\eps > 0$, $\eps' \geq 0$ and $\rhoAB \in \normstates{\hAB}$, then
  \begin{align*}
    &\chmineeps{\eps'}{A}{B}{\rho} - \log c \leq
    \chminalteeps{\eps+\eps'}{A}{B}{\rho} \leq
    \chmineeps{\eps+\eps'}{A}{B}{\rho} \, ,
  \end{align*}
  where $c = 2/\eps^2 + 1/(1 - \eps')$.
\end{lemma}
The equivalence of the max-entropies follows by their definition as
duals, i.e.\ we have
\begin{align*}
  \chmaxeeps{\eps'}{A}{B}{\rho} + \log c \geq
  \chmaxalteeps{\eps+\eps'}{A}{B}{\rho} \geq 
  \chmaxeeps{\eps+\eps'}{A}{B}{\rho} \, .    
\end{align*}
 
For convenience of exposition, we introduce the generalized
conditional min-entropy
\begin{align*}
 \chminrel{A}{B}{\rho|\sigma} := -D_{\textrm{max}}(\rhoAB\,\|\,\idA
 \kron \sigmaB)\, .
\end{align*}
The proof of Lemma~\ref{lemma:equivalence} is based on the following
result.
\begin{lemma}
\label{lemma:min-bound-1}
Let $\eps > 0$ and $\rhoABC \in \subnormstates{\hABC}$ be pure.
Then, there exists a projector $\PiAC$ on $\hi{AC}$ and a
state $\rhotABC = \PiAC\, \rhoABC\, \PiAC$ such that $\rhotABC \in
\epsballpure{\rhoABC}$ and
\begin{align*}
  \chminrel{A}{B}{\rhot|\rho} \geq \chmin{A}{B}{\rho} - \log
  \frac{2}{\eps^2} \, .
\end{align*}
Furthermore, there exists a state $\rhobAB \in \subnormstates{\hAB}$
that satisfies $\idx{\rhob}{AB} \in
\epsball{\rhoAB}$ and
\begin{align*}
  \chminalt{A}{B}{\rhob} \geq \chmin{A}{B}{\rho} - \log \Big(
  \frac{2}{\eps^2} + \frac{1}{\tr\,\rhoAB} \Big) \, .
\end{align*}
\end{lemma}
\begin{proof}
The proof is structured as follows: First, we give a lower bound on
the entropy $\chminrel{A}{B}{\rhot|\rho}$ in terms of
$\chmin{A}{B}{\rho}$ and a projector $\PiB$ that is the dual
projector (cf.\ Corollay~\ref{cor:mirror-project}) of $\PiAC$ with
regard to $\rhoABC$. We then find a lower bound on the purified
distance between $\rhoABC$ and $\rhotABC$ in terms of $\PiB$ and
define $\PiB$ (and, thus, $\PiAC$) such that this distance does not
exceed $\eps$.
%The second statement follows for the state $\rhobAB =
%\rhotAB + \tauA \kron (\rhoB - \rhotB)$, where $\tauA$ is the fully
%mixed state on $\hA$.

Let $\lambda$ and $\sigmaB$ be the pair that optimizes the min-entropy
$\chmin{A}{B}{\rho}$, i.e.\ $\chmin{A}{B}{\rho} =
\chminrel{A}{B}{\rho|\sigma} = -\log \lambda$.
We have $\rhotB \leq \rhoB$ by definition of $\rhotABC$. Hence,
$\chminrel{A}{B}{\rhot|\rho}$ is finite and can be written as
\begin{align*}
  2^{-\chminrel{A}{B}{\rhot|\rho}} &= \normbig{
    \rhoB^{-\nicefrac{1}{2}} \rhotAB \rhoB^{-\nicefrac{1}{2}}
  }{\infty} \, ,
\end{align*}
where $\norm{X}{\infty}$ denotes the maximum eigenvalue of $X$.
We bound this expression using the dual projector $\PiB$ of $\PiAC$
with regard to $\rhoABC$ and the fact that $\rhoAB \leq \lambda
\idA \kron \sigmaB$ by definition of $\lambda$ and $\sigmaB$:
\begin{align*}
  \textrm{rhs.} &= \normbig{\ptr{C} \big( (\PiAC \kron \rhoB^{-\nicefrac{1}{2}})\,
    \rhoABC\, (\PiAC \kron \rhoB^{-\nicefrac{1}{2}}) \big) }{\infty} \\
  &= \normbig{ \PiB\, \rhoB^{-\nicefrac{1}{2}} \rhoAB\,
    \rhoB^{-\nicefrac{1}{2}} \PiB }{\infty} \\
  &\leq \lambda\, \normbig{ \idA \kron \PiB\, \rhoB^{-\nicefrac{1}{2}}
    \sigmaB\, \rhoB^{-\nicefrac{1}{2}} \PiB }{\infty} \\
 &= \lambda \,\norm{\PiB \GammaB \PiB}{\infty} \, ,
\end{align*}
where, in the last step, we introduced the Hermitian operator
$\GammaB := \rhoB^{-\nicefrac{1}{2}} \sigmaB\,
\rhoB^{-\nicefrac{1}{2}}$. Taking the logarithm on both sides leads to
\begin{align}
  \label{eqn:min-bound-1}
  \chminrel{A}{B}{\rhot|\rho} \geq \chmin{A}{B}{\rho} - \log
  \norm{\PiB \GammaB \PiB}{\infty} \, .
\end{align}

We use Lemma~\ref{lemma:project-dist} to bound the distance between
$\rhoABC$ and $\rhotABC$, namely
\begin{align*}
  P(\rhoABC, \rhotABC) \leq \sqrt{2\, \trace{\PiAC^\perp \rhoABC}} =
  \sqrt{2\,\trace{\PiB^\perp \rhoB}} \, ,
\end{align*}
where the last equality can be verified using
Corollary~\ref{cor:mirror-project}.  Clearly, the optimal choice of
$\PiB$ will cut off the largest eigenvalues of $\GammaB$
in~\eqref{eqn:min-bound-1} while keeping the states $\rhoABC$ and
$\rhotABC$ close. We thus define $P_B$ to be the minimum rank
projector onto the smallest eigenvalues of $\GammaB$ such that
$\trace{\PiB \rhoB} \geq \tr\,\rhoB - \eps^2/2$ or,
equivalently, $\trace{\PiB^\perp \rhoB} \leq \eps^2/2$. This
definition immediately implies that $\rhoABC$ and $\rhotABC$ are
$\eps$-close and it remains to find an upper bound on
$\norm{\PiB\GammaB\PiB}{\infty}$.

Let $\PiB'$ be the projector onto the largest remaining eigenvalue in $\PiB
\GammaB \PiB$ and note that $\PiB'$ and $\PiB^\perp$ commute with $\GammaB$. Then,
\begin{align*}
  \norm{\PiB \GammaB \PiB}{\infty} = \trace{\PiB' \GammaB} =
  \min_{\muB} \frac{\trace{\muB (\PiB^\perp + \PiB)
      \GammaB}}{\trace{\muB}} \, ,
\end{align*}
where $\muB$ is minimized over all positive operators in the support
of $\PiB^\perp + \PiB'$. Fixing instead $\muB = (\PiB^\perp + \PiB') \rhoB
(\PiB^\perp + \PiB')$, we find
\begin{align*}
  \norm{\PiB \GammaB \PiB}{\infty} &\leq
  \frac{\trace{\GammaB^{\nicefrac{1}{2}} \rhoB
      \GammaB^{\nicefrac{1}{2}} (\PiB^\perp + \PiB')}}{\trace{(\PiB^\perp +
      \PiB') \rhoB}} \\
  &\leq \frac{\trace{\GammaB^{\nicefrac{1}{2}} \rhoB
      \GammaB^{\nicefrac{1}{2}}}}{\trace{(\PiB^\perp +
      \PiB') \rhoB}} \leq \frac{2}{\eps^2} \, . \\
\end{align*}
In the last step we used that $\trace{\rhoB^{\nicefrac{1}{2}}
  \GammaB \rhoB^{\nicefrac{1}{2}}} = \trace{\sigmaB} = 1$ and that
$\trace{(\PiB^\perp + \PiB')\rhoB} \geq \frac{\eps^2}{2}$ by
definition of $\PiB^\perp$. We have now established
the first statement.

To prove the second statement, we introduce an operator $\DeltaB :=
\rhoB - \rhotB \geq 0$. The state $\rhobAB = \rhotAB + \idA /
\idx{d}{A} \kron \DeltaB$, where $\idx{d}{A} = \dim \hA$, satisfies
$\rhobB = \rhoB$.  We now show that the state $\rhobAB$ is
$\eps$-close to $\rhoAB$. The inequality $\rhotAB \leq \rhobAB$
implies $\norm{\sqrt{\rhotAB} \sqrt{\rhoAB}}{1} \leq
\norm{\sqrt{\rhobAB} \sqrt{\rhoAB}}{1}$ and, thus,
\begin{align*}
  \bar{F}(\rhoAB, \rhobAB) &\geq F(\rhotAB, \rhoAB) + 1 - \tr\,\rhoAB \\
  &\geq F(\rhotABC, \rhoABC) + 1 - \tr\,\rhoAB \\
  &= 1 - \trace{\PiAC^\perp \rhoAC} \geq 1 - \eps^2/2 \, ,
\end{align*}
where we used the monotonicity of the fidelity $F(\rho, \tau) :=
\norm{\sqrt{\rho}\sqrt{\tau}}{1}$ under the partial
trace. Thus, $P(\rhobAB, \rhoAB) \leq \eps$.

We use that $\rhobB = \rhoB$ and $\rhobAB \leq \rhotAB + \idA/\idx{d}{A} \kron
\rhoB$ to find a lower bound on $\chminalt{A}{B}{\rhob} =
\chminrel{A}{B}{\rhob|\rho}$\,:
\begin{align*}
  2^{-\chminalt{A}{B}{\rhob}} &= \normbig{ \rhoB^{-\nicefrac{1}{2}} \rhobAB \, \rhoB^{-\nicefrac{1}{2}} }{\infty} \\
  &\leq \normBig{\rhoB^{-\nicefrac{1}{2}} \rhotAB \, \rhoB^{-\nicefrac{1}{2}} + \frac{1}{\idx{d}{A}} \idAB}{\infty} \\
  &\leq \lambda\, \frac{2}{\eps^2} + \frac{1}{\idx{d}{A}} \, .
\end{align*}
We have $\lambda \geq \tr\,\rhoAB / \idx{d}{A}$ (Lemma 20 in
\cite{tomamichel09}) and, thus,
\begin{align*}
  \chminalt{A}{B}{\rhob} \geq \chmin{A}{B}{\rho} - \log \Big(
  \frac{2}{\eps^2} + \frac{1}{\tr\,\rhoAB} \Big) \, .
\end{align*}
This concludes the proof of the second statement.
\end{proof}

Furthermore, the alternative smooth min-entropy is a lower bound on
the smooth min-entropy by definition.
\begin{lemma}
\label{lemma:min-bound-2}
Let $\rhoAB \in \subnormstates{\hAB}$, then
\begin{align*}
  \chminalt{A}{B}{\rho} \leq \chmin{A}{B}{\rho} - \log
  \frac{1}{\tr\,\rhoAB} \, .
\end{align*}
\end{lemma}

We are now ready to prove Lemma~\ref{lemma:equivalence}. Namely, we
show that, for $\eps > 0$, $\eps' \geq 0$ and $\rhoAB \in
\subnormstates{\hAB}$, it holds that
\begin{align*}
&\chmineeps{\eps'}{A}{B}{\rho} - \log c \leq
\chminalteeps{\eps+\eps'}{A}{B}{\rho} \leq
\chmineeps{\eps+\eps'}{A}{B}{\rho} \, ,
\end{align*}
where $c = 2/\eps^2 + 1/(\tr\,\rhoAB - \eps')$.
\begin{proof}[Proof of Lemma~\ref{lemma:equivalence}]
Let $\rhotAB \in
\eepsball{\eps'}{\rhoAB}$ be the state that maximizes
$\chmineeps{\eps'}{A}{B}{\rho}$. Clearly, $\tr\,\rhotAB \geq \tr\,\rhoAB -
\eps'$. Moreover, Lemma~\ref{lemma:min-bound-1} and the triangle
inequality of the purified distance imply that there exists a
state $\rhobAB \in \eepsball{\eps+\eps'}{\rhoAB}$ that satisfies
\begin{align*}
  \chminalteeps{\eps+\eps'}{A}{B}{\rho} \geq \chminalt{A}{B}{\rhob} \geq \chmineeps{\eps'}{A}{B}{\rho} - \log c
  \, ,
\end{align*}
which concludes the proof of the first inequality. The second
inequality follows by applying Lemma~\ref{lemma:min-bound-2} to the
state that maximizes $\chminalteeps{\eps+\eps'}{A}{B}{\rho}$.
\end{proof}

\section{Collision Entropy}
\label{app:coll}

In this section, we prove Lemma~\ref{lemma:coll-bound}, which gives a
relation between the collision entropy and the min-entropy. First, we
provide an inequality in terms of relative entropies.

\begin{lemma}
\label{lemma:min-coll-bound}
Let $\rhoAB \in \subnormstates{\hAB}$ and $\sigmaB \in \normstates{\hB}$, then
\begin{align*}
  D_{\textnormal{max}}(\rhoAB\,\|\,\idA \kron \sigmaB) \geq \log
  \idx{\Gamma}{C}(\rhoAB|\sigmaB) - \log \tr\,\rhoAB \, .
\end{align*}
\end{lemma}
\begin{proof}
By definition of the max relative entropy, we have
$\rhoAB \leq 2^{D_{\textnormal{max}}(\rhoAB\|\idA \kron \sigmaB)} \idA \kron \sigmaB$
and, thus,
\begin{align*}
  (\idA \kron \sigmaB^{-\nicefrac{1}{2}}) \rhoAB  (\idA \kron
  \sigmaB^{-\nicefrac{1}{2}}) \leq
  2^{D_{\textnormal{max}}(\rhoAB\|\idA \kron \sigmaB)}\, \idAB \, .
\end{align*}
We use this and the fact that $\trace{\rhoAB X} \leq \trace{\rhoAB Y}$
if $X \leq Y$ to get $$\idx{\Gamma}{C}(A|B)_{\rho|\sigma} \leq
2^{D_{\textnormal{max}}(\rhoAB\,\|\,\idA \kron \sigmaB)}\,
\tr\,\rhoAB\, ,$$ which concludes the proof.
\end{proof}

Using the above result and Lemma~\ref{lemma:min-bound-1} of
Appendix~\ref{app:alt}, we are ready to prove
Lemma~\ref{lemma:coll-bound} of Section~\ref{sec:def}.
\begin{proof}[Proof of Lemma~\ref{lemma:coll-bound}]
  To prove the first statement, we apply
  Lemma~\ref{lemma:min-coll-bound} to the state $\rhoXB$. The
  inequality holds in particular for the state $\sigmaB$ that
  optimizes $\chmin{X}{B}{\rho}$ (cf.~Definition~\ref{def:min}),
  establishing~\eqref{eqn:min-coll-bound}.

  Next, we use Lemma~\ref{lemma:min-bound-1} to define $\rhobXB \in
  \eepsball{\bar{\eps}}{\rhoXB}$. Thus,
  \begin{align*}
    \chminalt{X}{B}{\rhob} \geq \chmin{X}{B}{\rho} - \log \Big(
    \frac{2}{\bar{\eps}^2} + \frac{1}{\tr\,\rhoXB} \Big) \, .
  \end{align*}
  In particular, we can choose $\rhobXB$ normalized and
  CQ.\footnote{To see this, first note that the alternative
    min-entropy, $\chminalt{X}{B}{\rhob}$, is independent of
    $\tr\,\rhobXB$. Moreover, measuring $\rhobXB$ on the X system will
    increase the alternative min-entropy while the distance to
    $\rhoXB$ can only decrease.}  We apply
  Lemma~\ref{lemma:min-coll-bound} to this state to get
  \begin{align*}
    \idx{\Gamma}{C}(\rhobXB | \rhobB) \leq
    2^{-\chminalt{X}{B}{\rhob}} \leq
    2^{-\chmin{X}{B}{\rho} + \log (\frac{2}{\bar{\eps}^2} + 1)} \, ,
  \end{align*}
  which concludes the proof of~\eqref{eqn:minalt-coll-bound}.
\end{proof}

\section{Duality Relation for Alternative Smooth Entropies}
\label{app:altmax}

Here, we find that the alternative smooth min-entropy of A
conditioned on B is invariant under local isometries on the B system.
Since all purifications are equivalent up to isometries on the
purifying system, this allows the definition of the alternative
max-entropy as its dual (see Definition~\ref{def:max}). Furthermore,
the max-entropy of A conditioned on B is invariant under local
isometries on the B system as a direct consequence. Note that the
alternative smooth min- and max-entropies are in general not invariant
under isometries on the A system, i.e.\ they depend on the dimension
of the Hilbert space $\hA$.
%The proof follows similar lines
%as the proof of the same property for the smooth min-entropy
%$\chmineps{A}{B}{}$ (cf.~\cite{tomamichel09}).

\begin{lemma}
\label{lemma:altiso}
Let $\eps \geq 0$ and $\rhoAB \in \subnormstates{\hAB}$. Moreover,
let $U: \hB \to \hD$ be an isometry with $\idx{\tau}{AD} := (\idA
\kron U) \rhoAB (\idA \kron U^\dagger)$. Then,
\begin{align*}
  \chminalteps{A}{B}{\rho} &= \chminalteps{A}{D}{\tau}
  \quad \textrm{and}\\
  \chmaxalteps{A}{B}{\rho} &= \chmaxalteps{A}{D}{\tau}\, .
\end{align*}
\end{lemma}
\begin{proof}
Let $\rhotAB \in \epsball{\rhoAB}$ be the state that maximizes the
alternative min-entropy of A conditioned on B and let $\lambda$ be
defined with $\chminalteps{A}{B}{\rho} = -\log \lambda$. Then
$\rhotAB \leq \lambda \idA \kron \rhotB$, which implies
$$ \underbrace{(\idA \kron U) \rhotAB (\idA \kron U^\dagger)}_{=:\,
  \idx{\taut}{AD}} \leq \lambda \idA \kron (U \rhotB U^\dagger) \,
.$$
Hence, $\idx{\taut}{AD} \leq \lambda \idA \kron
\idx{\taut}{D}$. Moreover, $\idx{\taut}{AD} \in
\epsball{\idx{\tau}{AD}}$ due to~\eqref{eqn:purified-monotone},
which implies $\chminalteps{A}{D}{\rho} \geq
\chminalteps{A}{B}{\rho}$.
The same argument in reverse can be
applied to get $\chminalteps{A}{B}{\rho} \geq \chminalteps{A}{D}{\tau}$.

The invariance under isometry of the dual quantity follows by
definition. Namely, let $\rhoABE$ be any purification of $\rhoAB$, then
\begin{align*}
  \chmaxalteps{A}{B}{\rho} &= -\chminalteps{A}{E}{\rho} \\
  &= -\chminalteps{A}{E}{\tau} = \chmaxalteps{A}{D}{\tau} \, ,
\end{align*}
where $\idx{\tau}{ADE} := (\idA \kron U \kron \idE) \rhoABE (\idA
\kron U^\dagger \kron \idE)$ is a purification of $\idx{\tau}{AD}$.
\end{proof}

Next, we derive expression~\eqref{eqn:altmax} for the alternative
non-smooth and smooth max-entropies.
The result for the non-smooth entropy was first shown in~\cite{berta08}
and an alternative proof is provided here for completeness.
\begin{lemma}
Let $\rhoAB \in \subnormstates{\hAB}$, then
\begin{align*}
  \chmaxalt{A}{B}{\rho} = \!\!\max_{\sigmaB \in \normstates{\hB}}\!
  \log \Trace{ \Pi_{\rhoAB} (\idA \kron \sigmaB) }
\end{align*}
\end{lemma}
\begin{proof}
Let $\rhoABC$ be a purification of
$\rhoAB$. Then, $\tauABC := \big(\idAB \kron
\rhoC^{-\nicefrac{1}{2}} \big) \rhoABC\, \big( \idAB \kron
\rhoC^{-\nicefrac{1}{2}} \big)$ has marginal $\tauAB = \Pi_{\rhoAB}$
due to Lemma~\ref{lemma:mirror}. This allows us to write
\begin{align*}
  2^{\chmaxalt{A}{B}{\rho}} &= 2^{-\chminalt{A}{C}{\rho}} 
%    &= \normbig{ (\idA \kron \rhoC^{-\nicefrac{1}{2}})
%      \rhoAC\, (\idA \kron \rhoC^{-\nicefrac{1}{2}}) }{\infty} \\
  = \norm{\tauAC}{\infty}
  = \norm{\tauB}{\infty} \\
  &= \max_{\sigmaB} \trace{ \sigmaB \tauB } = \max_{\sigmaB} \Trace{
    \Pi_{\rhoAB} (\idA \kron \sigmaB) } ,
\end{align*}
where the maximization is over all $\sigmaB \in
\normstates{\hB}$.
\end{proof}
The alternative smooth max-entropy can be seen as an optimization of
the non-smooth quantity over an $\eps$-ball of states, where the ball
is embedded in a sufficiently large Hilbert space. We show
that~\eqref{eqn:altmaxeps} holds.
\begin{lemma}
Let $\eps \geq 0$ and $\rhoAB \in \subnormstates{\hAB}$, then
\begin{align*}
  \chmaxalteps{A}{B}{\rho} &= \!\! \inf_{\hi{B'} \supseteq
  \hB} \ \min_{\idx{\rhot}{AB'} \in
  \epsball{\idx{\rho}{AB'}}}\! \chmaxalt{A}{B'}{\rhot}\, ,
\end{align*}
where $\idx{\rho}{AB'}$ is the embedding of $\rhoAB$ into
$\hi{AB'}$. Furthermore, the infimum is attained for embeddings with
$\dim \hi{B'} \geq \dim \supp{\rhoAB} \cdot \dim \hA$.
\end{lemma}
\begin{proof}
Let $\rhoABC$ be a purification of $\rhoAB$ on a Hilbert space $\hC$
with $\dim \hC = \rank\,\{\rhoAB\}$. Furthermore, for any
$\hi{B'} \supseteq \hB$, let $\idx{\rho}{AB'C'}$ be the
embedding of $\rhoABC$ into $\hi{AB'C'}$ with $\dim \hi{C'} =
\dim \hi{AB'}$. We use Corollary~\ref{cor:pure-ball} twice to
upper bound
\begin{align*}
  \chmaxalteps{A}{B}{\rho} &=- \chminalteps{A}{C'}{\rho} \\
  &= \min_{\idx{\rhot}{AC'} \in \epsball{\idx{\rho}{AC'}}}
  -\chminalt{A}{C'}{\rhot}\\
  &\leq \min_{\idx{\rhot}{AB'C'} \in \epsballpure{\idx{\rho}{AB'C'}}}
  \chmaxalt{A}{B'}{\rhot}\\
  &= \min_{\idx{\rhot}{AB'} \in \epsball{\idx{\rho}{AB'}}}
  \chmaxalt{A}{B'}{\rhot} \, .
\end{align*}

A lower bound on $\chmaxalteps{A}{B}{\rho}$ follows when we
require that $\dim \hi{B'} \geq \rank\,\{\rhoAB\} \cdot \dim \hA
= \dim \hAC$. Then, $\hi{B'}$ is large enough to
accomodate all purifications of states in $\hAC$. Using
Corollay~\ref{cor:pure-ball} twice, we find
\begin{align*}
  \chmaxalteps{A}{B}{\rho} &= \min_{\rhotAC \in \epsball{\eps}{\rhoAC}}
  -\chminalt{A}{C}{\rhot}\\
  &= \min_{\idx{\rhot}{AB'C} \in
    \epsballpure{\idx{\rho}{AB'C}}} \chmaxalt{A}{B'}{\rhot}\\
  &\geq \min_{\idx{\rhot}{AB'} \in
    \epsball{\idx{\rho}{AB'}}} \chmaxalt{A}{B'}{\rhot} \, .
\end{align*}
The infimum is therefore attained and it is sufficient to consider
embeddings with $\dim \hi{B'} = \dim
\supp{\rhoAB} \cdot \dim \hA$.
\end{proof}


\end{document}